\documentclass[pdftex,prx,aps,nofootinbib,floatfix,superscriptaddress]{revtex4-2}
\usepackage[utf8]{inputenc}
\usepackage{braket}
\usepackage{amsmath}
\usepackage{geometry}
\usepackage{amsfonts}
\usepackage{graphicx}
\usepackage[T1]{fontenc}
\usepackage{hyperref}
\usepackage{amsthm}
\usepackage{amssymb}
\newtheorem{definition}{Definition}[section]
\newtheorem{example}{Example}[section]
\newtheorem{claim}{Claim}[section]
\newtheorem{theorem}{Theorem}[section]
\newtheorem{conjecture}{Conjecture}[section]
\newtheorem{lemma}{Lemma}[section]
\usepackage{vwcol}
\usepackage{tikz}
\usepackage{mathtools}
\usepackage{tasks}
\usepackage{appendix}
\usepackage{algorithm}
\usepackage{algpseudocode}
\usepackage{makecell}
\usepackage{wrapfig}

\newcommand{\bigsigma}{\raisebox{-.35\baselineskip}{\huge\ensuremath{\Sigma}}}

\begin{document}

\begin{abstract}
The quantum Schur transform is a fundamental building block that maps the computational basis to a coupled basis consisting of irreducible representations of the unitary and symmetric groups. 
Equivalently, it may be regarded as a change of basis from the computational basis to a simultaneous spin eigenbasis of Permutational Quantum Computing (PQC) [Quantum Inf. Comput., 10, 470–497 (2010)]. By adopting the latter perspective, we present a transparent algorithm for implementing the qubit quantum Schur transform as a unitary operation which uses $O(\log(n))$ ancillas and can be decomposed into a sequence of $O(n^3\log(n)\log(\frac{n}{\epsilon}))$ Clifford + T gates, where $\epsilon$ is the accuracy of the algorithm in terms of the trace norm. We discuss the necessity for some applications of implementing this operation as a unitary rather than an isometry, as is often presented.

By studying the associated Schur states, which consist of qubits coupled via Clebsch-Gordan coefficients, we introduce the notion of generally coupled quantum states. We present six conditions, which in different combinations ensure the efficient preparation of these states on a quantum computer or their classical simulability (in the sense of computational tractability). It is shown that Wigner 6-j symbols and SU(N) Clebsch-Gordan coefficients naturally fit our framework. Finally, we investigate unitary transformations which preserve the class of computationally tractable states. 

\end{abstract}

\title{Generalised Coupling and An Elementary Algorithm for the Quantum Schur Transform}
\author{Adam Wills}
\author{Sergii Strelchuk}
\affiliation{DAMTP, Centre for Mathematical Sciences, University of Cambridge, Cambridge CB30WA, UK}

\maketitle

\section{Introduction}

Identifying the ingredients responsible for quantum  advantage is one of the major challenges in quantum computing. Proofs of unconditional separation between certain complexity classes are few but include, notably,~\cite{bravyi2018quantum,watts2023unconditional}. On the other end, the study of classical simulability of quantum computational processes provides a way to explore the boundary between classical and quantum computational power. Two of the most important examples of classically efficiently simulatable computations come from the Gottesmann-Knill theorem \cite{gottesman1997stabilizer} and the simulation of nearest neighbour matchgate circuits \cite{jozsa2008matchgates}.
 In certain cases, the ability to classically simulate quantum computational processes comes as a surprise -- as in the case of Permutational Quantum Computing (PQC)~\cite{jordan2009permutational,havlicek2018quantum, HavlicekThesis}. The latter represents a restricted class of quantum computations originating in the works of Marzuoli, Rasetti and Penrose \cite{marzuoli2005computing,penrose1971angular} and in its current form formalised by Jordan in \cite{jordan2009permutational}. This class of quantum computations was subsequently shown to be classically efficiently simulatable, first in the Schur basis in \cite{havlicek2018quantum} and later generalised to a wider range of bases in \cite{HavlicekThesis}. PQC operates in a state-space of coupled qubits, where the couplings are  SU(2) Clebsch-Gordan coefficients. The proofs of its classical simulability all rely on the fact that such coefficients satisfy relations that correspond to the conservation of angular momentum. This property enables one to show that the states involved in PQC belong to the broad class of computationally tractable (CT) states -- a very general notion of classical simulation first introduced by Van den Nest in~\cite{VdNProbMethods}. 

The centerpiece of PQC (in the Schur basis) is the quantum Schur transform -- this couples computational basis states and forms arguably the only non-trivial component of the computational process. It finds a number of applications in machine learning~\cite{zheng2022super}, physics~\cite{harrow2005applications}, quantum information~\cite{cirac1999optimal}, chemistry~\cite{pauncz1967alternant} and many others~\cite{keyl2001estimating, knill2000theory,kempe2001theory,hayashi2002optimal}.
In recent years, interest in the Schur transform has led to a number of implementations~\cite{bacon2005quantum,kirby2017practical, krovi2019efficient}. In our work, we present a new, elementary algorithm for this operation.
It will be presented on qubits, but the algorithm readily extends to the qudit case. Furthermore, it also generalises to the other unitaries utilised in PQC (we will refer to them as `PQC unitaries'). Before discussing the algorithms for implementing these more general unitaries, we will first discuss the fascinating connection between Schur-Weyl duality and PQC.
After this, we will introduce the notion of generally coupled states and study their properties. In this setting, Clebsch-Gordan coefficients will be replaced by general `coupling coefficients' that can couple any constant number of systems. Sufficient conditions will be given for these coefficients to define physical quantum states, followed by conditions that make these states efficiently preparable on a quantum computer. Finally, we will introduce sufficient conditions for these states to be computationally tractable. The one condition that will bring these states from efficient preparation on a quantum computer to computational tractability will turn out to be (a very broad notion of) conservation of angular momentum. We will show two examples of coefficients that naturally fit our formalism: the SU(N) Clebsch-Gordan coefficients and the Wigner 6-j symbols. The SU(N) Clebsch-Gordan coefficients are the components of the transformation mapping the basis defined by the tensor product of two SU(N) irreps to that of the corresponding direct sum decomposition -- a widespread operation readily finding use in representation theory. Meanwhile, the Wigner 6-j symbols arise in the recoupling theory of SU(2), and so are in fact relevant in PQC, but find applications far beyond - most notably in the study of spin networks ~\cite{yutsis1962mathematical, varshalovich1988quantum}, with applications to quantum gravity \cite{haggard2011asymptotic} as well as in calculations for molecular scattering \cite{de2003orthogonal}.

\section{Outline of Results}

The paper is structured as follows. In Section \ref{PQCUnitaries}, we introduce the Schur transform from the spin eigenbasis point of view which naturally arises in the context of Permutational Quantum Computing (PQC). We then show how to implement the Schur transform on qubits in an elegant way. This is arguably the simplest algorithm realising the Schur transform as a unitary operation to date and uses $O(n\log(n))$ ancillas. In Appendix \ref{cleanTransform}, we will discuss why emphasis must be placed on the unitarity of the operation for some applications. In Section \ref{logManyAncillas}, we show how a reduction to $O(\log(n))$ ancillas may be achieved. In this section, we will also present exact counts on the number of ancillas used in our original and modified algorithms for the Schur transform. In Section \ref{SchurWeyl}, we discuss the connection between Schur-Weyl duality and PQC, and show how all PQC unitaries can be performed in Section \ref{genPQCUnitaries}. All of these algorithms will have their run times analysed in Appendix \ref{runtime}.

In Section \ref{genCoupling}, we rigorously define the general notion of coupling of quantum systems, where PQC couplings arise as a special case. We give sufficient conditions for the corresponding states to be computationally tractable -- a widely adopted notion of classical simulability -- as well as efficiently preparable on a quantum computer. Finally, in Section \ref{CTgates}, we discuss the subtleties of the notion of computational tractability and the difficulties in defining a sound notion of unitary gates which preserve it, followed by several results and a conjecture about such gates. 

\section{PQC Unitaries}\label{PQCUnitaries}

We first introduce the quantum Schur transform in the context of Permutational Quantum Computing (PQC) -- a restricted model of quantum computation, developed in its current form by Jordan in \cite{jordan2009permutational} with earlier ideas from~\cite{marzuoli2005computing}.  One of the key quantities we will repeatedly use is angular momentum:

\begin{definition}
On $n$ qubits, let $\left(\sigma_x^{(i)}, \sigma_y^{(i)}, \sigma_z^{(i)}\right)$ denote the usual Pauli operators acting on the i-th qubit. The angular momentum operator on the i-th qubit is then defined to be

\begin{equation}
    \vec{S}^{(i)} = \frac{1}{2}\begin{pmatrix}\sigma_x^{(i)} \\ \sigma_y^{(i)} \\ \sigma_z^{(i)} \end{pmatrix}.
\end{equation}
Given a subset of the qubits $a \subseteq \{1, ..., n\}$, we may then define the total angular momentum of this subset as $S^2_a = \left(\sum_{i \in a}\vec{S}^{(i)}\right)\cdot\left(\sum_{i \in a}\vec{S}^{(i)}\right)$, as well as the Z-angular momentum of the subset as $Z_a = \frac{1}{2}\sum_{i \in a}\sigma_z^{(i)}$.
\end{definition}

\noindent Eigenvalues of total angular momentum operators are often referred to as $j$-values, while those of Z-angular momentum operators are commonly referred to as $m$-values. A PQC basis is then a simultaneous eigenbasis of a particular collection of $n$ of the above operators that pairwise commute. The allowed collections of operators are described after Claim \ref{commutability}. We refer to a state from a PQC basis as a PQC state. The results that we will need concerning the commutability of the above operators are detailed below. 

\begin{claim}For non-empty subsets of the qubits a and b,\label{commutability}

\begin{itemize}
    \item If $a \cap b = \emptyset$ then $[S^2_a, S^2_b]$ = $0$.
    \item If $a \subseteq b$ then $[S^2_a, S^2_b]$ = $0$.
    \item If $a \subseteq b$ then $[S^2_a, Z_b]$ = $0$.
\end{itemize}
\end{claim}
\noindent The first point is immediate because the operators act on disjoint systems. The proofs of the latter two points can be found in the appendix of \cite{SchurSampling}. Now consider a collection of distinct, non-empty, proper subsets of the qubits $a_1, ..., a_{n-2}$ such that for each pair $a_i$ and $a_j$, $i \neq j$, $a_i \subseteq a_j$, $a_j \subseteq a_i$ or $a_i \cap a_j = \emptyset$. Given the above results, we find that the $n$ operators $(S^2_{a_1}, ..., S^2_{a_{n-2}}, S^2, Z)$ are pairwise commuting, where $S^2 \coloneqq S^2_{\{1, ..., n\}}$ and $Z \coloneqq Z_{\{1, ..., n\}}$\footnote{Eigenvalues of $S^2_a$, where $a \neq \{1, ..., n\}$ are often called internal $j$-values while the eigenvalues of the other two operators are often called the root $J$ and $M$-values, as appropriate. By a small abuse of notation, the root $J$ and $M$-values themselves are sometimes referred to as $S^2$ and $Z$, but strictly $S^2$ and $Z$ are actually operators.}. Being Hermitian operators, these may be diagonalised, resulting in a simultaneous eigenbasis \cite{jordan2009permutational}. Given such a choice of subsets $a_1, ..., a_{n-2}$, the induced simultaneous eigenbasis is called a PQC basis. 

A useful way of representing the above bases, and the states they contain, is via rooted binary trees, illustrated in the following example. PQC bases are in one-to-one correspondence with unlabelled binary trees and PQC states from a given PQC basis are in one-to-one correspondence with the labellings of the corresponding unlabelled tree \cite{jordan2009permutational}.
\begin{example}
Consider the following labelled binary tree.
\\
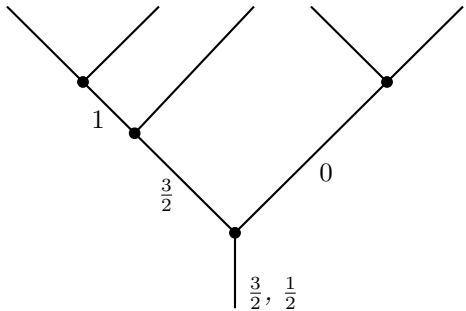
\begin{center}
\begin{tikzpicture}

\draw[fill=black] (3,1) circle (2pt);
\draw[fill=black] (1.68,2.32) circle (2pt);
\draw[fill=black] (5,3) circle (2pt);
\draw[fill=black] (1,3) circle (2pt);

\draw[thick] (3,0) -- (3,1) -- (1.68,2.32) -- (1,3) -- (0,4);
\draw[thick] (1,3) -- (2,4);
\draw[thick] (1.68,2.32) -- (3.25,4);
\draw[thick] (5,3) -- (4,4);
\draw[thick] (3,1) -- (5,3);
\draw[thick] (5,3) -- (6,4);

\node at (1.2,2.5) {$1$};
\node at (2.1, 1.5) {$\frac{3}{2}$};
\node at (4.2,1.8) {$0$};
\node at (3.5,0.2) {$\frac{3}{2}$, $\frac{1}{2}$};

\end{tikzpicture}
\end{center}

\noindent If one removes the labels from this tree, the obtained unlabelled binary tree is a diagrammatic representation of the simultaneous eigenbasis of the operators $\left(S^2_{\{1, 2\}}, S^2_{\{1, 2, 3\}}, S^2_{\{4, 5\}}, S^2, Z\right)$ in the space of five qubits. With the labels as shown, this is a diagrammatic representation of one of the 32 states from that basis. Here, the first two qubits are in a spin-1 state, the first three are in a spin-3/2 state, etc. At the root, the eigenvalues of $S^2$ and $Z$ are shown respectively. Note that the angular momentum values must obey the usual laws under addition of systems, noting that individual qubits are considered as having spin $\frac{1}{2}$, since they are two-state systems. To find the wavefunction of such a state, we need only look at the structure of the tree and use Clebsch-Gordan coefficients, $C^{J, M}_{j_1, m_1; j_2, m_2}$. This state has wavefunction

\begin{equation}
\sum_{(x_1, ..., x_5) \in \{-\frac{1}{2}, \frac{1}{2}\}^5} C^{3/2, 1/2}_{3/2, x_1 + x_2 + x_3; 0, x_4 + x_5} C^{3/2, x_1 + x_2 + x_3}_{1, x_1 + x_2; 1/2, x_3}C^{1, x_1 + x_2}_{1/2, x_1; 1/2, x_2}C^{0,x_4 + x_5}_{1/2, x_4, 1/2, x_5} \ket{x_1 ... x_5}.
\end{equation}
Throughout the paper it will be convenient to write computational basis states as $\ket{\pm \frac{1}{2}}$, simply so we may most easily refer to the $x_i$ as angular momenta. For clarity, we mean $\ket{-\frac{1}{2}} = \ket{\frac{1}{2},-\frac{1}{2}} = \ket{0}$ and $\ket{\frac{1}{2}} = \ket{\frac{1}{2},\frac{1}{2}} = \ket{1}$, written in our conventions, usual angular momentum notation and the usual computational basis respectively. The Clebsch-Gordan coefficients are coupling quantum systems together by assigning one coefficient per vertex. We will look at more general ways to couple quantum systems with this as the prototypical example in Section \ref{genCoupling}. Part of that will be to generalise the notion of conservation of angular momentum - embodied here in the relation $m_1 + m_2 \neq M \implies C^{J, M}_{j_1, m_1; j_2, m_2} = 0$.
\end{example}

\noindent Special emphasis is placed on the PQC basis defined by the operators $\left(S^2_{\{1, 2\}}, S^2_{\{1, 2, 3\}}, ..., S^2_{\{1, ..., n-1\}}, S^2, Z\right)$. This is called the Schur basis, or sometimes sequentially coupled basis. The unitary operation mapping the computational basis to the Schur basis is called the Schur Transform\footnote{Note that this is quite a physical way to define the Schur transform. An equivalent mathematical formulation exists and can be found in, for example, \cite{bacon2005quantum}.}, while we refer to unitaries mapping the computational basis to a generic PQC basis (a basis with an arbitrary choice of $a_1, ..., a_{n-2}$ satisfying the aforementioned conditions) as PQC unitaries.

A single instance of PQC is the following sequence of steps:

\begin{enumerate}
    \item Prepare a PQC state from some PQC basis.
    \item Apply a permutation to the qubits.
    \item Measure in some PQC basis.
\end{enumerate}

\noindent Thus, a single PQC instance is specified by a labelled binary tree, a qubit permutation and an unlabelled binary tree. As explained in \cite{jordan2009permutational}, by making polynomially many of these measurements, one can estimate the output probabilities to within $\pm \epsilon$ using $O\left(1/\epsilon^2\right)$ iterations by the usual means. The complexity class of problems solvable by a PQC, PQP, is then defined in \cite{jordan2009permutational} to be the problems that may be solved by estimating output probabilities of PQC instances. We will distinguish two PQC models: weak PQC -- the model of computing with access to these probabilities to polynomial precision, and strong PQC, defined as having access to the matrix elements themselves (physically, this would correspond to the computer being able to perform the Hadamard test \cite{aharonov2006polynomial}, for example) i.e. $\bra{\psi}U_{\pi}\ket{\phi}$ for $\ket{\psi}$ and $\ket{\phi}$ PQC states and $U_\pi$ a qubit permutation. Both strong and weak models were shown to be classically, efficiently simulatable in~\cite{havlicek2018quantum} and \cite{HavlicekThesis}. We briefly recall the proof structure:

\begin{itemize}
    \item All PQC states are computationally tractable (CT, as defined in \cite{VdNProbMethods}). This was shown for Schur states in~\cite{havlicek2018quantum} and later for all PQC states in \cite{HavlicekThesis}.
    
    \item Permutations map CT states to CT states (this is discussed in \cite{SchurSampling}, for example).
    
    \item We can classically, efficiently compute the overlap between two CT states to polynomial accuracy, as shown in \cite{VdNProbMethods}.
\end{itemize}

\noindent Despite the classical simulability of PQC, Schur states remain important quantum objects to study, especially in quantum information theoretic protocols \cite{keyl2001estimating, knill2000theory, kempe2001theory, hayashi2002optimal}. More generally, it is unknown how to efficiently simulate a generic operation on such a state, rather than simply a permutation\footnote{It is known how to simulate the action of a CT-preserving gate (like a permutation) on a Schur state - see Section \ref{CTgates} - but simulating the action of any other gate seems to necessitate a quantum computer.}, and so it is desirable to be able to prepare these states on a quantum computer. There are different approaches to implementing the Schur transform on a quantum computer \cite{krovi2019efficient, kirby2017practical, bacon2005quantum}, but none of them are able to implement general PQC unitaries. Furthermore, the algorithms of \cite{kirby2017practical} and \cite{bacon2005quantum} map Schur states to computational basis states that include ancillas that encode those Schur states - essentially the $j$ and $m$-values mentioned above. These are implementations of the Schur transform as an \textit{isometry}. While this may be acceptable in some applications, it is problematic in others, for example if one were to attempt to implement the algorithm of \cite{zheng2022super}, for which the unitary Schur transform is essential. The differences between differing notions of the Schur transform, and the necessity of the unitary operation for some applications, are discussed further in Appendix \ref{cleanTransform}. The remarkable algorithm of \cite{krovi2019efficient} provides an efficient unitary transform of $n$-qudit computational basis states to Schur states, although it relies on a complex representation-theoretic treatment of the problem which poses significant obstacles when applying it in practice\footnote{We also mention that \cite{bacon2005quantum}, along with Harrow's method found in \cite{harrow2005applications} to reduce the time scaling in qudit dimension $d$ to $O(log\;d)$, does perform the desired operation, but again uses somewhat more complex theory.}.

One of the emphases of this work is the implementation of a unitary Schur transform in simple terms along with the provision of an explicit gate count. This enables us to generalise it to implement general PQC unitaries mapping  $n$-qubit computational basis states to PQC states, including Schur states.
This will require no representation theory to understand - only the basic knowledge of SU(2) Clebsch-Gordan coefficients (although the reader may find the brief introduction to Schur-Weyl duality of Section \ref{SchurWeyl} enriching). While we will derive our results for qubit systems, our algorithms can be easily generalised to work with qudits much like all the prior implementations. We will omit the analysis of run time for qudits.

\subsection{The Schur Transform}\label{schurAlgoSection}

We now turn to our algorithm for the implementation for the Schur transform on qubits which, as explained above, is the unitary mapping the computational basis to the PQC basis described by the unlabelled binary tree on $n$ qubits

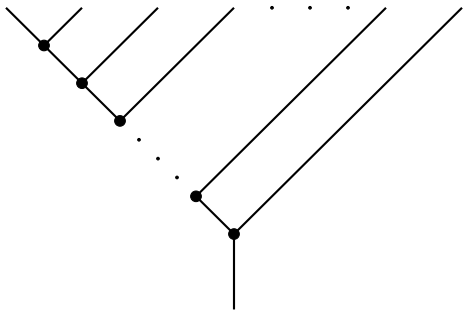
\begin{center}

\begin{tikzpicture}

\draw[fill=black] (0,1) circle (2pt);
\draw[fill=black] (-0.5,1.5) circle (2pt);
\draw[fill = black] (-0.75,1.75) circle (0.5pt);
\draw[fill = black] (-1,2) circle (0.5pt);
\draw[fill = black] (-1.25,2.25) circle (0.5pt);
\draw[fill = black] (-1.5,2.5) circle (2pt);
\draw[fill = black] (0.5,4) circle (0.5pt);
\draw[fill = black] (1,4) circle (0.5pt);
\draw[fill = black] (1.5,4) circle (0.5pt);
\draw[fill = black] (-2,3) circle (2pt);
\draw[fill = black] (-2.5,3.5) circle (2pt);

\draw[thick] (0,0) -- (0,1) -- (3,4);
\draw[thick] (0,1) -- (-0.5,1.5) -- (2,4);
\draw[thick] (-1.5,2.5) -- (0,4);
\draw[thick] (-1.5,2.5) -- (-2,3) -- (-1,4);
\draw[thick] (-2,3) -- (-2.5,3.5) -- (-2,4);
\draw[thick] (-2.5,3.5) -- (-3,4);

\end{tikzpicture}
\end{center}

\noindent which, in particular, contains $n-1$ nodes. Note that we do not specify any particular ordering of the Schur states i.e. we do not specify which computational basis states map to which Schur states. Our algorithm can work for any ordering by modifying the relevant mappings. We will present the algorithm as if some desired ordering has been fixed. First, we will introduce a compact notation for Schur states. $\ket{j_1, j_2, ..., j_{n-2},J,M}$ refers to the Schur state with an eigenvalue of $j_1$ for the operator $S^2_{\{1, 2\}}$, an eigenvalue $j_2$ for the operator $S^2_{\{1, 2, 3\}}$ and so on, up to an eigenvalue $j_{n-2}$ for the operator $S^2_{\{1, 2, 3, ..., n-1\}}$, an eigenvalue $J$ for the operator $S^2$ and an eigenvalue $M$ for the operator $Z$.

Our algorithm will works in two stages. The first stage is called the pre-mapping. Given a computational basis state on $n$ qubits, $\ket{x}$, where $x \in B_n$, suppose $\ket{x}$ is mapped to $\ket{j_1, j_2, ..., j_{n-2}, J, M}$. The pre-mapping achieves the operation

\begin{equation}
    \ket{x} \mapsto \ket{j_1}\ket{j_2} ... \ket{j_{n-2}}\ket{J}\ket{M}
\end{equation}

\noindent coherently over all computational basis states $\ket{x}$. The right-hand side is just another computational basis state, but it contains ancillary qubits (for simplicity of notation these ancillary qubits are not shown on the left-hand side). The above mapping requires $O(n\;log(n))$ ancillary qubits, however in Section \ref{logManyAncillas} we will show how to reduce this number down to $O(log(n))$. The computational basis state on the right-hand side simply encodes the eigenvalues into each of its $n$ registers. For example, $j_3$ can take values $0$, $1$ or $2$ and so two qubits are required for this register and any encoding may be chosen, say $00$ for $0$, $01$ for $1$ and $10$ for $2$. Using the usual rules of addition of angular momentum, the maximum possible value of $j_i$ is $\frac{i + 1}{2}$, the maximum possible value of $J$ is $\frac{n}{2}$ and $M$ ranges in integer steps between $-J$ and $J$. This means that each variable can take a number of values polynomial (indeed, linear) in $n$, and so each register contains logarithmically many qubits.

The second stage is the coupling itself. The pre-mapped state $\ket{j_1}\ket{j_2} ... \ket{j_{n-2}}\ket{J}\ket{M}$ will be mapped to the corresponding Schur state, 

\begin{equation}
    \sum_{x \in B_n} C^{j_1, x_1 + x_2}_{1/2, x_1; 1/2, x_2} C^{j_2, x_1 + x_2 + x_3}_{j_1, x_1 + x_2; 1/2, x_3} ... C^{J, M}_{j_{n-2}, x_1 + ... + x_{n-1}; 1/2, x_n} \ket{x} \label{SchurState}
\end{equation}

\noindent which is a superposition of $n$-qubit computational basis states forming exactly the desired Schur state\footnote{As before, we find it convenient to denote computational basis states with $\pm \frac{1}{2}$ rather than $0$ and $1$, so strictly $B_n$ is the set $\left\{ \pm \frac{1}{2}\right\}^n$ in this context.}. Note that this is performed coherently over all pre-mapped states and note also that the final state is one on $n$ qubits only, so in this mapping the ancillary qubits have become un-entangled (and discarded).

We will now show how to perform the pre-mapping. The pre-mapping procedure mirrors the coupling of qubits one-by-one, fixing the registers $j_1$, $j_2$ and so on, down to $M$. In the first step, we fix the register $\ket{j_1}$. To do this requires only a mapping on the first two qubits (with ancillary qubits being introduced as needed),

\begin{equation}
    \ket{x_1 x_2} \mapsto \ket{j_1}\ket{m_1}.
\end{equation}

\noindent The above action is essentially the pre-mapping that we would do if we were performing the Schur transform on only 2 qubits. $m_1$ represents the $z$-angular momentum of the first two qubits and is encoded into computational basis states of as many qubits as needed, just as for the other variables. We need only pick some ordering, say $00 \mapsto j_1 = 1, m_1 = 1$, $01 \mapsto j_1 = 1, m_1 = 0$, $10 \mapsto j_1 = 1, m_1 = -1$ and $11 \mapsto j_1 = 0, m_1 = 0$\footnote{It is in making these decisions about the ordering that we can decide to which Schur state each computational basis state maps.}. The next stage of the pre-mapping then takes place on three registers:

\begin{equation}
    \ket{j_1}\ket{m_1}\ket{x_3} \mapsto \ket{j_1}\ket{j_2}\ket{m_2}
\end{equation}

\noindent where, similarly, $m_2$ represents the $Z$-angular momentum of the first three qubits. This, again, is done just by picking some ordering, say $j_1 = 1, m_1 = 1, x_3 = 1/2 \mapsto j_1 = 1, j_2 = 3/2, m_2 = 3/2$ etc., where notice that $j_1$ does not change (but must be acted on as part of this operation effectively as a control register). We continue in this way, where all of the operations except the first take place on three registers. After $i$ steps of this, we have, in total, 

\begin{equation}
    \ket{j_1}\ket{j_2}...\ket{j_i}\ket{m_i}\ket{x_{i+2}}\ket{x_{i+3}}...\ket{x_n}
\end{equation}

\noindent for $1 \leq i \leq n-2$. Note that each of these acts on three (except for the first, which acts on two) registers of logarithmically many qubits and they all have classically, efficiently computable matrix elements (just correctly placed ones and zeros) so all of these can be efficiently implemented. After $n-2$ of these, we perform the final map, which is very much the same but with different labels:

\begin{equation}
    \ket{j_1}\ket{j_2}...\ket{j_{n-2}}\ket{m_{n-2}}\ket{x_n} \mapsto \ket{j_1}\ket{j_2}...\ket{j_{n-2}}\ket{J}\ket{M}.
\end{equation}

\noindent We therefore perform $n-1$ efficiently implementable operations to complete the pre-mapping, which is therefore efficiently implementable.

\begin{algorithm}[H]
	\caption{Pre-Mapping Stage for the Schur Transform\label{premapSchur}}
	\textbf{Input:}  Any superposition of $n$-qubit computational basis states: $\sum_xc_x\ket{x}$.\\
\textbf{Output:} The same superposition with each computational basis state mapped to its corresponding Schur encoding on $O(n\;log(n))$ qubits: $\sum_xc_x\ket{j_1(x)}\ket{j_2(x)}...\ket{j_{n-2}(x)}\ket{J(x)}\ket{M(x)}$.
	\begin{algorithmic}[1]
\State{Map the first two qubits to the registers encoding $j_1$ and $m_1$ in computational basis states, using some ordering: $\ket{x_1x_2} \mapsto \ket{j_1}\ket{m_1}$.}
			\For{$i=1$ to $n-3$}
		    \State{Perform the mapping between computational basis states on three registers $\ket{j_i}\ket{m_i}\ket{x_{i+2}} \mapsto \ket{j_i}\ket{j_{i+1}}\ket{m_{i+1}}$ using some ordering.}
			      \EndFor
		\State{Perform the mapping between computational basis states on three registers $\ket{j_{n-2}}\ket{m_{n-2}}\ket{x_n} \mapsto \ket{j_{n-2}}\ket{J}\ket{M}$ using some ordering.}
		\end{algorithmic}
\end{algorithm}

\noindent The coupling stage will then work in the opposite direction; where the pre-mapping progressed down the tree, the coupling stage moves up the tree, adding the Clebsch-Gordan coefficient relevant to each vertex at each of the $n-1$ steps of this stage. Each step will be efficiently implementable for the same reason as before - we act on logarithmically many qubits at a time, and our matrix elements are classically, efficiently computable (Clebsch-Gordan coefficients can be classically, efficiently computed when their entries are of absolute value $poly(n)$~\cite{havlicek2018quantum}, which they are here as mentioned earlier). The first step acts on the last three registers as

\begin{align}
    \ket{j_1}\ket{j_2}...\ket{j_{n-2}}\ket{J}\ket{M} &\mapsto \sum_{m_{n-2}, x_n}\ket{j_1} \ket{j_2} ... \ket{j_{n-3}}C^{J,M}_{j_{n-2}, m_{n-2}; 1/2, x_n}\ket{j_{n-2}}\ket{m_{n-2}}\ket{x_n}.
\end{align}

\noindent Note that in making this operation, the ancillary qubits on the last register have all become un-entangled.
Further note that this operation, and all the operations in this stage, are unitary due to the following orthogonality property of Clebsch-Gordan coefficients \cite{CGWiki}:

\begin{equation}
    \sum_{m_1, m_2} C^{J,M}_{j_1, m_1; j_2, m_2} C^{J', M'}_{j_1, m_1; j_2, m_2} = \delta_{J J'} \delta_{M M'}
\end{equation}

\noindent where we note that Clebsch-Gordan coefficients are real, at least in the most common phase convention that we adopt -- the Condon-Shortley convention. By acting on three registers at a time, and two registers for the last step, we eventually end up with

\begin{equation}
    \sum_{m_{n-2}, x_n} ... \sum_{m_1, x_3} \sum_{x_1, x_2} C^{j_1, m_1}_{1/2, x_1; 1/2, x_2}C^{j_2, m_2}_{j_1, m_1; 1/2, x_3} ... C^{J,M}_{j_{n-2}, m_{n-2}; 1/2, x_n}\ket{x_1 x_2 ... x_n}
\end{equation}

\noindent which is exactly the desired Schur state in Equation \eqref{SchurState} given that Clebsch-Gordan coefficients preserve angular momentum ($C^{J,M}_{j_1, m_1; j_2, m_2} = 0$ unless $m_1 + m_2 = M$) and so all of the sums over $m_i$ can be made to vanish. All of our operations are performed coherently over computational basis state superpositions, and so the pre-mapping and coupling stages together implement the Schur transform.

\begin{algorithm}[H]
	\caption{Coupling Stage of the Schur Transform\label{couplingSchur}}
\textbf{Input:} Any superposition of computational basis states on $O(n\:log(n))$ qubits encoding the Schur states: $\sum_xc_x\ket{j_1(x)}\ket{j_2(x)}...\ket{j_{n-2}(x)}\ket{J(x)}\ket{M(x)}$.\\
\textbf{Output:} The same superposition with each computational basis state encoding mapped to its corresponding Schur state on $n$ qubits.
	\begin{algorithmic}[1]
\State{Perform the mapping on the last three registers: $\ket{j_{n-2}}\ket{J}\ket{M} \mapsto \newline \sum_{m_{n-2}, x_n}C^{J,M}_{j_{n-2}, m_{n-2}; 1/2, x_n}\ket{j_{n-2}}\ket{m_{n-2}}\ket{x_n}$.}
			\For{$i=n-3$ to 1}
		    \State{Perform the mapping on three registers $\ket{j_i}\ket{j_{i+1}}\ket{m_{i+1}} \mapsto \sum_{m_{i}, x_{i+2}}C^{j_{i+1}, m_{i+1}}_{j_i, m_i; 1/2, x_{i+2}}\ket{j_i}\ket{m_i}\ket{x_{i+2}}$.}
			      \EndFor
		\State{Perform the mapping on two registers $\ket{j_1}\ket{m_1} \mapsto \sum_{x_1, x_2}C^{j_1, m_1}_{1/2, x_1; 1/2, x_2}\ket{x_1x_2}$.}
		\end{algorithmic}
\end{algorithm}

\noindent An analysis of the run time of the whole algorithm for the Schur transform can be found in Appendix~\ref{runtime}, as well as a run time analysis for the preparation of individual Schur states.

Let us discuss briefly why the pre-mapping is a necessary step for a clean and efficient transform. Without it, we are coupling qubits directly. This could be attempted moving down the tree, or up the tree. Attempting to do this directly going down the tree does not lead to an efficient implementation:

\begin{align}
\ket{x} &\mapsto \sum_{\mathclap{\substack{x_1, x_2}}} C^{j_1, x_1 + x_2}_{1/2, x_1; 1/2, x_2} \ket{x}\\
&\mapsto \sum_{\mathclap{\substack{x_1, x_2, x_3}}} C^{j_1, x_1 + x_2}_{1/2, x_1; 1/2, x_2}C^{j_2, x_1 + x_2 + x_3}_{j_1, x_1 + x_2; 1/2, x_3} \ket{x} \mapsto ...
\end{align}

\noindent as we see that by the end, we will be acting simultaneously on all $n$ qubits. It is possible, however, to put Clebsch-Gordan coefficients on the registers, moving down the tree, by going to an encoding as follows:

\begin{align}
\ket{x} &\mapsto \sum_{\mathclap{\substack{j_1, m_1}}}C^{j_1, m_1}_{1/2, x_1; 1/2, x_2}\ket{j_1}\ket{m_1}\ket{x_3...x_n}\\
&\mapsto \sum_{\mathclap{\substack{j_1, j_2\\m_1, m_2}}}C^{j_1, m_1}_{1/2, x_1; 1/2, x_2}C^{j_2, m_2}_{1/2, m_1; 1/2, x_3}\ket{j_1}\ket{j_2}\ket{m_2}\ket{x_4...x_n}
\end{align}

\noindent and so on, where in a general step, we will act on the registers $\ket{j_i}\ket{m_i}\ket{x_{i+2}}$. This does produce an efficient algorithm but ultimately gives us

\begin{equation}
    \sum_{\mathclap{\substack{j_1, ..., j_{n-2}, J, M}}}C^{j_1, x_1 + x_2}_{1/2, x_1; 1/2, x_2}...C^{j_{n-2}, x_1 + ... + x_{n-1}}_{j_{n-3}, x_1 + ... + x_{n-2}; 1/2, x_{n-1}}C^{J,M}_{j_{n-2}, x_1 + ... + x_{n-1}; 1/2, x_n}\ket{j_1}...\ket{j_{n-2}}\ket{J}\ket{M}
\end{equation}

\noindent which is much alike the inverse Schur transform, except it leads to a superposition of encoded states\footnote{Indeed, if we ran the inverse of our pre-mapping stage, then this would give us exactly the inverse of our algorithm.} and it is not easy to see how to compute on these. This is what the main algorithm of \cite{bacon2005quantum} performs, except with mathematical notation on the encoding (and note also, this paper presents several improvements upon this, such as reducing the spatial overhead to logarithmic). \cite{kirby2017practical} performs something similar, but with an alternative encoding.

We now ask what occurs if we attempt to simply couple the qubits from the bottom of the tree upwards, without any pre-mapping at the start. The issue here is that, given some computational basis states $\ket{x}$, it is hard to determine from $\ket{x}$ alone what the eigenvalues corresponding to $\ket{x}$ are at the bottom of the tree i.e. in order to perform, say, 

\begin{equation}
    \ket{x} \mapsto \sum_{\mathclap{\substack{m_{n-2}, x_n}}}C^{J,M}_{j_{n-2}, m_{n-2}, 1/2, x_n}\ket{x_1...x_{n-3}}\ket{j_{n-2}}\ket{m_{n-2}}\ket{x_n}
\end{equation}

\noindent we must act on all the qubits at once to work out the $J, M$ and $j_{n-2}$ that arise from each $x$. Conversely, we find that with the pre-mapping, we can work out the $j_1$ (and $m_1$) arising from each $x$ only by looking at the first two qubits; we can find out the correct $j_2$ from $j_1$ (and $m_1$) and the third qubit, and so on. Together, this is what makes the pre-mapping and coupling stages work as they do.

\subsection{PQC Trees and Schur-Weyl Duality}\label{SchurWeyl}

Before shifting to general PQC unitaries, we briefly discuss the PQC bases in the context of Schur-Weyl duality. We will aim to provide a fresh perspective on Schur-Weyl duality through the lens of PQC. For the sake of generality, we will give background in terms of general qudits, although in most other places in this paper, our attention will be limited to qubits. Further representation-theoretic details can be found in~\cite{bacon2005quantum}. Consider the Hilbert space of $n$ qudits: $V^{\otimes n}$, where $V \cong \mathbb{C}^d$. Schur-Weyl duality is a statement about two representations on this space given by the maps $\mathbf{P}$ and $\mathbf{Q}$, to use the same notation as \cite{bacon2005quantum}. The former gives a representation of the permutation group $S_n$, acting on computational basis states as

\begin{equation}
    \mathbf{P}(\sigma)\ket{i_1\; ... \;i_n} = \ket{i_{\sigma^{-1}(1)}\; ...\; i_{\sigma^{-1}(n)}}
\end{equation}

\noindent for $\sigma \in S_n$. As such, $\mathbf{P}(\sigma)$ simply rearranges the qudits according to $\sigma$. $\mathbf{Q}$ provides a representation of $d$-dimensional unitaries: $U_d$, also acting in a fairly simple way:
\begin{equation}
    \mathbf{Q}(U)\ket{i_1\; ...\; i_n} = \left(U\ket{i_1}\right) \otimes ... \otimes \left(U\ket{i_n}\right)
\end{equation}

\noindent where $U \in U_d$; $U$ is just applied to each individual system. Schur-Weyl duality is the statement that $\mathbf{P}$ and $\mathbf{Q}$ may be simultaneously decomposed into a direct sum of irreps i.e. there exists a basis in which both $\mathbf{P}$ and $\mathbf{Q}$ are simultaneously block-diagonalised. The set of distinct irreps into which $\mathbf{P}$ and $\mathbf{Q}$ decompose are commonly written $\{\mathbf{p}_\lambda\}$ and $\{\mathbf{q}_\lambda^d\}$ respectively, but note that in each decomposition, each irrep may come with some multiplicity. The irrep spaces for $\mathbf{p}_\lambda$ and $\mathbf{q}_\lambda^d$ may be written as $\mathcal{P}_\lambda$ and $\mathcal{Q}_\lambda^d$ respectively. The situation may be summarised in the equation

\begin{equation}
    \left(\mathbb{C}^d\right)^{\otimes n} \overset{U_d \times S_n}{\cong} \bigoplus_\lambda \mathcal{Q}_\lambda^d \otimes \mathcal{P}_\lambda.
\end{equation}

\noindent This equation tells us that our whole space may be decomposed into a direct sum of spaces, $\mathcal{Q}_\lambda^d \otimes \mathcal{P}_\lambda$. Each of these spaces may be viewed either as $dim(\mathcal{Q}_\lambda^d)$ copies of the irrep $\mathbf{p}_\lambda$ or as $dim(\mathcal{P}_\lambda)$ copies of the irrep $\mathbf{q}_\lambda^d$. We are therefore told that in the simultaneous decompositions of $\mathbf{P}$ and $\mathbf{Q}$, the dimension of $\mathbf{p}_\lambda$ equals the multiplicity of $\mathbf{q}_\lambda^d$ and the dimension of $\mathbf{q}_\lambda^d$ equals the multiplicity of $\mathbf{p}_\lambda$.

We now turn to Schur-Weyl duality in the context of PQC, where from now on we specialise to qubits only. It turns out that PQC bases provide the bases needed for the above described block diagonalisation. The representation $\mathbf{P}$ is dealt with in \cite{jordan2009permutational}, where it is pointed out that both $S^2$ and $Z$ commute with $\mathbf{P}(\sigma) \; \forall \; \sigma$. However, the simultaneous eigenspaces of $S^2$ and $Z$ form irrep spaces of $S_n$ \cite{pauncz1967alternant}. Therefore, if we take the set of labelled PQC trees of a given shape with given $S^2$ and $Z$ root labels, we obtain an orthonormal basis for some copy of the irrep $\mathbf{p}_\lambda$. Note in particular that by choosing a different shape for the tree but keeping the same $S^2$ and $Z$ root labels, we obtain a different basis for the same irrep space. A given PQC basis therefore forms an orthonormal basis for the whole space $\left(\mathbb{C}^2\right)^{\otimes n}$ in which the action of $\mathbf{P}$ is block diagonalised. 

The action of $\mathbf{Q}$ is then block diagonalised in these bases also. We can see this as follows. Local unitaries like $\mathbf{Q}(U)$ do not affect total angular momentum of any subset of the qubits. Therefore, any subspace spanned by a set of these basis vectors with fixed $j$-labels (both internal and on the root), for which only the root $M$-label is allowed to vary, is invariant under the action of $\mathbf{Q}$. We know from Schur-Weyl duality that the multiplicity space of $\mathbf{p}_\lambda$ is the representation space of $\mathbf{q}_\lambda^2$, and so we find that this subspace in fact forms an irrep space of $\mathbf{q}_\lambda^2$. 

With this in mind, we wish to emphasise the very intuitive perspective that this provides for us on Schur-Weyl duality. Summing over the irrep label $\lambda$ corresponds to summing over the root $J$-value. With such a label fixed, the irrep $\mathcal{P}_\lambda$ is formed by fixing a root $M$-value and allowing the internal $j$-labels to vary, whereas the irrep $\mathcal{Q}_\lambda$ is formed by fixing internal $j$-labels and allowing the root $M$-label to vary. From this, it is clear why the multiplicity of $\mathbf{p}_\lambda$ equals the dimension of $\mathbf{q}_\lambda^2$ and the multiplicity of $\mathbf{q}_\lambda^2$ equals the dimension of $\mathbf{p}_\lambda$. Moreover, it is also physically natural that these spaces form irrep spaces. In particular, the action of $\mathbf{Q}(U)$ has an analogue in classical angular momentum. Collectively re-orienting a system of bodies all in the same way does not change their combined angular momentum - and so analogously the $j$-value of all qubit subsets remains fixed. The orientation of these classical objects has changed, though, and analogously the $M$-label at the root of our tree may be caused to vary.

\begin{example}

The following shows all of the PQC states on four qubits of the displayed shape with $S^2 = 1$ where, again, a root label of $a,b$ refers to $S^2 = a, Z = b$.
\\
\\

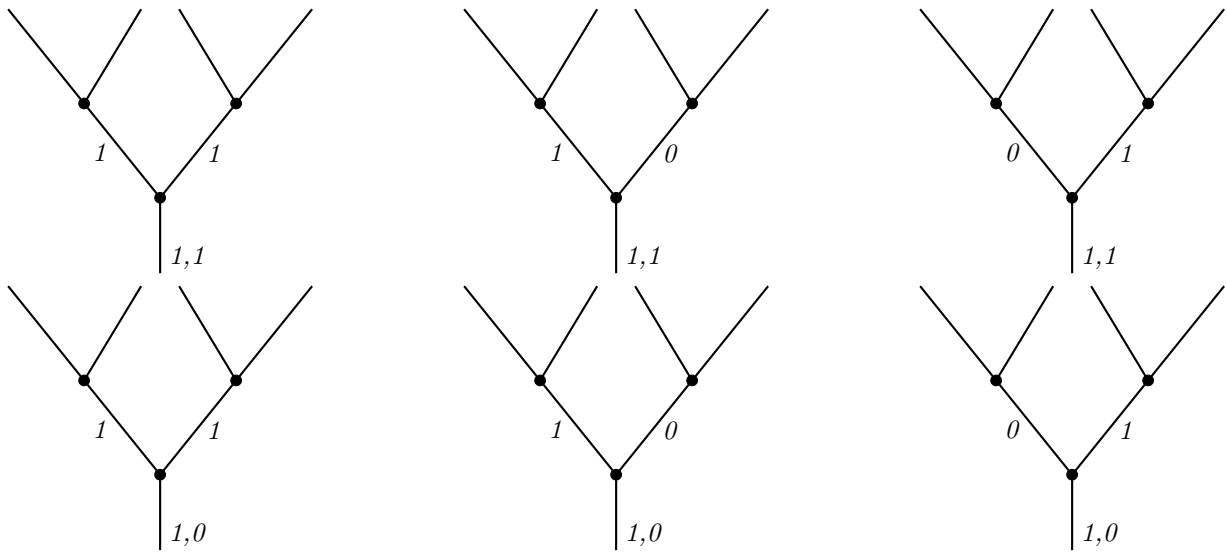
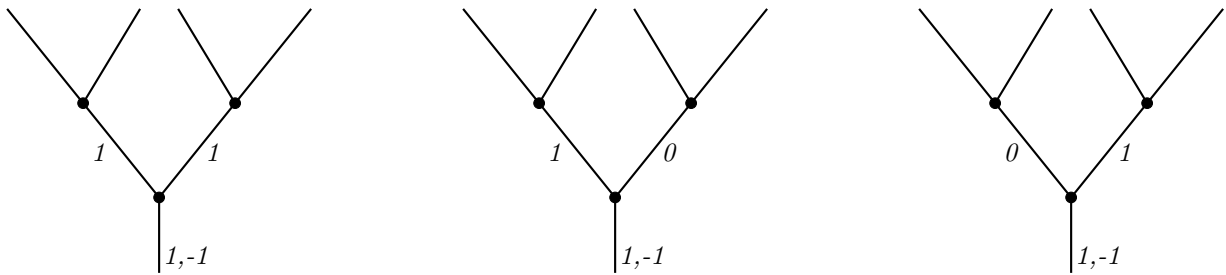
\begin{tikzpicture}

\draw[fill=black] (2-6,1) circle (2pt);
\draw[fill=black] (1-6,2.25) circle (2pt);
\draw[fill=black] (3-6,2.25) circle (2pt);

\node at (1.2-6,1.6) {1};
\node at (2.7-6,1.59) {1};
\node at (2.35-6,0.2) {1,1};

\draw[thick] (2-6,0) -- (2-6,1) -- (1-6,2.25) -- (0-6,3.5);
\draw[thick] (1-6,2.25) -- (1.75-6,3.5);
\draw[thick] (3-6,2.25) -- (2.25-6,3.5);
\draw[thick] (2-6,1) -- (3-6,2.25) -- (4-6,3.5);

\draw[fill=black] (2,1) circle (2pt);
\draw[fill=black] (1,2.25) circle (2pt);
\draw[fill=black] (3,2.25) circle (2pt);

\node at (1.2,1.6) {1};
\node at (2.7,1.59) {0};
\node at (2.35,0.2) {1,1};

\draw[thick] (2,0) -- (2,1) -- (1,2.25) -- (0,3.5);
\draw[thick] (1,2.25) -- (1.75,3.5);
\draw[thick] (3,2.25) -- (2.25,3.5);
\draw[thick] (2,1) -- (3,2.25) -- (4,3.5);

\draw[fill=black] (2+6,1) circle (2pt);
\draw[fill=black] (1+6,2.25) circle (2pt);
\draw[fill=black] (3+6,2.25) circle (2pt);

\node at (1.2+6,1.6) {0};
\node at (2.7+6,1.59) {1};
\node at (2.35+6,0.2) {1,1};

\draw[thick] (2+6,0) -- (2+6,1) -- (1+6,2.25) -- (0+6,3.5);
\draw[thick] (1+6,2.25) -- (1.75+6,3.5);
\draw[thick] (3+6,2.25) -- (2.25+6,3.5);
\draw[thick] (2+6,1) -- (3+6,2.25) -- (4+6,3.5);
\end{tikzpicture}

\begin{tikzpicture}

\draw[fill=black] (2-6,1) circle (2pt);
\draw[fill=black] (1-6,2.25) circle (2pt);
\draw[fill=black] (3-6,2.25) circle (2pt);

\node at (1.2-6,1.6) {1};
\node at (2.7-6,1.59) {1};
\node at (2.35-6,0.2) {1,0};

\draw[thick] (2-6,0) -- (2-6,1) -- (1-6,2.25) -- (0-6,3.5);
\draw[thick] (1-6,2.25) -- (1.75-6,3.5);
\draw[thick] (3-6,2.25) -- (2.25-6,3.5);
\draw[thick] (2-6,1) -- (3-6,2.25) -- (4-6,3.5);

\draw[fill=black] (2,1) circle (2pt);
\draw[fill=black] (1,2.25) circle (2pt);
\draw[fill=black] (3,2.25) circle (2pt);

\node at (1.2,1.6) {1};
\node at (2.7,1.59) {0};
\node at (2.35,0.2) {1,0};

\draw[thick] (2,0) -- (2,1) -- (1,2.25) -- (0,3.5);
\draw[thick] (1,2.25) -- (1.75,3.5);
\draw[thick] (3,2.25) -- (2.25,3.5);
\draw[thick] (2,1) -- (3,2.25) -- (4,3.5);

\draw[fill=black] (2+6,1) circle (2pt);
\draw[fill=black] (1+6,2.25) circle (2pt);
\draw[fill=black] (3+6,2.25) circle (2pt);

\node at (1.2+6,1.6) {0};
\node at (2.7+6,1.59) {1};
\node at (2.35+6,0.2) {1,0};

\draw[thick] (2+6,0) -- (2+6,1) -- (1+6,2.25) -- (0+6,3.5);
\draw[thick] (1+6,2.25) -- (1.75+6,3.5);
\draw[thick] (3+6,2.25) -- (2.25+6,3.5);
\draw[thick] (2+6,1) -- (3+6,2.25) -- (4+6,3.5);
\end{tikzpicture}

\begin{tikzpicture}

\draw[fill=black] (2-6,1) circle (2pt);
\draw[fill=black] (1-6,2.25) circle (2pt);
\draw[fill=black] (3-6,2.25) circle (2pt);

\node at (1.2-6,1.6) {1};
\node at (2.7-6,1.59) {1};
\node at (2.35-6,0.2) {1,-1};

\draw[thick] (2-6,0) -- (2-6,1) -- (1-6,2.25) -- (0-6,3.5);
\draw[thick] (1-6,2.25) -- (1.75-6,3.5);
\draw[thick] (3-6,2.25) -- (2.25-6,3.5);
\draw[thick] (2-6,1) -- (3-6,2.25) -- (4-6,3.5);

\draw[fill=black] (2,1) circle (2pt);
\draw[fill=black] (1,2.25) circle (2pt);
\draw[fill=black] (3,2.25) circle (2pt);

\node at (1.2,1.6) {1};
\node at (2.7,1.59) {0};
\node at (2.35,0.2) {1,-1};

\draw[thick] (2,0) -- (2,1) -- (1,2.25) -- (0,3.5);
\draw[thick] (1,2.25) -- (1.75,3.5);
\draw[thick] (3,2.25) -- (2.25,3.5);
\draw[thick] (2,1) -- (3,2.25) -- (4,3.5);

\draw[fill=black] (2+6,1) circle (2pt);
\draw[fill=black] (1+6,2.25) circle (2pt);
\draw[fill=black] (3+6,2.25) circle (2pt);

\node at (1.2+6,1.6) {0};
\node at (2.7+6,1.59) {1};
\node at (2.35+6,0.2) {1,-1};

\draw[thick] (2+6,0) -- (2+6,1) -- (1+6,2.25) -- (0+6,3.5);
\draw[thick] (1+6,2.25) -- (1.75+6,3.5);
\draw[thick] (3+6,2.25) -- (2.25+6,3.5);
\draw[thick] (2+6,1) -- (3+6,2.25) -- (4+6,3.5);
\end{tikzpicture}
\vspace{0.5cm}

\noindent Each column forms a three dimensional irrep of $U_2$ and each row forms a three dimensional irrep of $S_4$. The rest of the 16 states in this PQC basis may be found by doing the same as the above for $S^2$ = 2 and $S^2$ = $0$. We note that all Schur-Weyl duality enforces in general is $dim(\mathbf{p}_\lambda) = mult(\mathbf{q}_\lambda^d)$ and $dim(\mathbf{q}_\lambda^d) = mult(\mathbf{p}_\lambda)$ but in this case all four values happen to coincide. For concreteness, we find the $U_2$ irrep, in the same way as it is done in \cite{kirby2017practical}. We may write a general unitary acting on a qubit as $U = \begin{pmatrix} a & -b^* \\ b & a^* \end{pmatrix}$ where $a, b \in \mathbb{C}$ satisfy $|a|^2 + |b|^2 = 1$. Recall that $\mathbf{Q}(U)$ acts on the whole space as $\mathbf{Q}(U) = U^{\otimes 4}$. To compute the action of this in the above basis, we then compute $U_{PQC}\mathbf{Q}(U)U_{PQC}^\dagger$, where $U_{PQC}$ is used to denote the PQC unitary mapping the computational basis to the PQC basis of the above shape. The result is a block-diagonal matrix with one five-dimensional block (corresponding to $S^2$ = 2), three three-dimensional blocks (as above) and two one-dimensional blocks (corresponding to $S^2$ = 0). The three-dimensional blocks that correspond to the above irreps are all equal to 

\begin{equation}
    \begin{pmatrix} a^2 & -\sqrt{2}ab^* & (b^*)^2 \\
    \sqrt{2}ab & |a|^2 - |b|^2 & -\sqrt{2}a^*b^* \\
    b^2 & \sqrt{2}a^*b & (a^*)^2
    \end{pmatrix}.
\end{equation}
\end{example}

\noindent PQC bases all vary by the different subgroups of $S_n$ to which they are adapted. The concept of subgroup adaptation is explained in detail in \cite{bacon2005quantum} and also in a way more relevant to this setting in \cite{jordan2009permutational}. As explained in \cite{jordan2009permutational}, suppose we fix a basis for an irrep of a group G, thereby obtaining some collection of matrices. If we restrict the inputs to the representation to some subgroup $H \leqslant G$, we will obtain some representation of H. The basis is called adapted to the subgroup H if the matrices map the elements of H to direct sums of irreducible representations of H i.e. the representation matrices are all block diagonal and each block forms an irrep of H. In some sense, a basis having this property means that it respects the structure of the group.

As noted in \cite{jordan2009permutational}, the PQC bases are adapted to certain $S_n$ subgroups with respect to the representation $\mathbf{P}$. If we fix a node in a PQC tree and consider the subgroup of $S_n$ that only permutes the qubits joining above that node, leaving all others fixed, then the PQC basis will be adapted to that subgroup, and this is true of every node. In particular, note that the Schur basis is adapted to the chain of subgroups $S_2 \leqslant S_3 \leqslant ... \leqslant S_{n-1} \leqslant S_n$, where $S_i$ is the subgroup permuting only the leftmost i qubits, leaving the remainder fixed\footnote{We note that, for the Schur Transform, our discussion for qubits is a simplification of the full story that applies to qudits. In that case, there is a similar subgroup tower corresponding to subgroups of $U_d$ \cite{bacon2005quantum}.}. The following example will elucidate this subgroup adaptation in the context of PQC bases and the representation $\mathbf{P}$.

\begin{example}
Consider the PQC basis on 6 qubits described by the following unlabelled tree:
\\

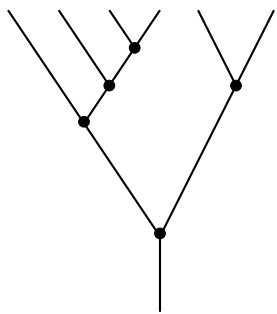
\begin{center}
\begin{tikzpicture}

\draw[fill=black] (2,1.04) circle (2pt);
\draw[fill=black] (1,2.52) circle (2pt);
\draw[fill=black] (3,3) circle (2pt);
\draw[fill=black] (1.333,3) circle (2pt);
\draw[fill=black] (1.667,3.5) circle (2pt);

\draw[thick] (2,0) -- (2,1) -- (1,2.5) -- (0,4);
\draw[thick] (1,2.5) -- (2,4);
\draw[thick] (2,1) -- (3,3) -- (3.5,4);
\draw[thick] (3,3) -- (2.5,4);
\draw[thick] (1.333,3) -- (0.667,4);
\draw[thick] (1.667,3.5) -- (1.333,4);
\end{tikzpicture}
\end{center}

\noindent We may find the block diagonal structure of the representation $\mathbf{P}$ in this basis by considering $U_{PQC}\mathbf{P}(\sigma)U_{PQC}^\dagger \; \forall \sigma \in S_n$ where now $U_{PQC}$ is the PQC unitary mapping the computational basis to the above PQC basis. We may also see an example of the subgroup adaptivity of this basis by considering $U_{PQC}\mathbf{P}(\sigma)U_{PQC}^\dagger$ where $\sigma$ is restricted to the subgroup of $S_n$ that permutes only the leftmost four qubits and leaves the rightmost two qubits fixed, because this tree has a vertex that joins exactly the leftmost four qubits.

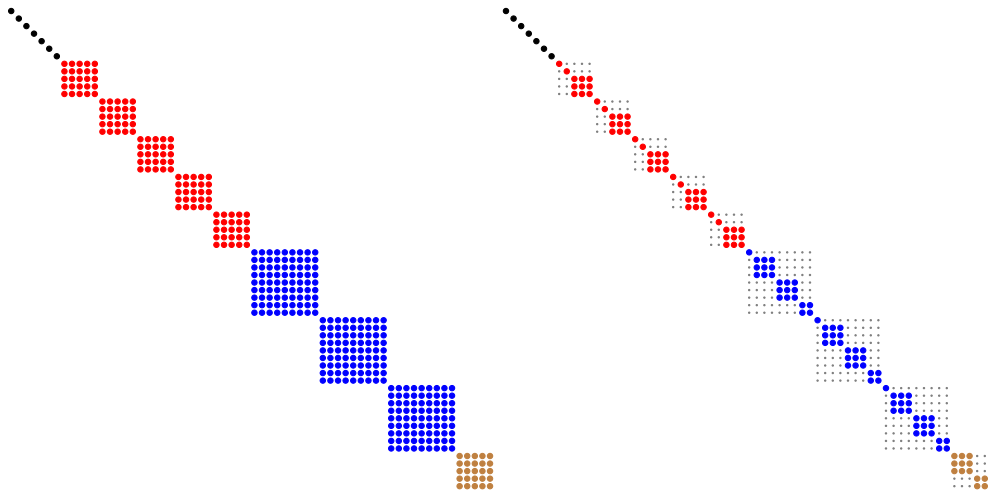
\begin{center}
\begin{tikzpicture}
\foreach \x in {1, ..., 7}
{
    \draw[fill=black] (0.1*\x, -0.1*\x) circle (1pt);
}
\foreach \i in {1, ..., 5}{
    \foreach \x in {1, ..., 5}{
        \foreach \y in {1, ..., 5}{
            \draw[red,fill=red] (0.2+0.1*\x+0.5*\i, -0.2-0.1*\y-0.5*\i) circle (1pt);
        }
    }
}
\foreach \i in {1, ..., 3}{
    \foreach \x in {1, ..., 9}{
        \foreach \y in {1, ..., 9}{
            \draw[blue,fill=blue] (2.3+0.1*\x+0.9*\i, -2.3-0.1*\y-0.9*\i) circle (1pt);
        }
    }
}
\foreach \x in {1, ..., 5}{
    \foreach \y in {1, ..., 5}{
        \draw[brown,fill=brown] (5.9+0.1*\x, -5.9-0.1*\y) circle (1pt);
    }
}
\end{tikzpicture}
\begin{tikzpicture}
\foreach \x in {1, ..., 7}
{
    \draw[fill=black] (0.1*\x, -0.1*\x) circle (1pt);
}
\foreach \i in {1, ..., 5}{
    \foreach \x in {1, ..., 5}{
        \foreach \y in {1, ..., 5}{
            \draw[gray,fill=gray] (0.2+0.1*\x+0.5*\i, -0.2-0.1*\y-0.5*\i) circle (0.25pt);
        }
    }
}
\foreach \i in {1, ..., 5}{
    \foreach \x in {1, ..., 2}{
        \draw[red,fill=red] (0.2+0.5*\i+0.1*\x,-0.2-0.5*\i-0.1*\x) circle (1pt);
    }
    \foreach \x in {3, ..., 5}{
        \foreach \y in {3, ..., 5}{
            \draw[red,fill=red] (0.2+0.1*\x+0.5*\i, -0.2-0.1*\y-0.5*\i) circle (1pt);
        }
    }
}
\foreach \i in {1, ..., 3}{
    \foreach \x in {1, ..., 9}{
        \foreach \y in {1, ..., 9}{
            \draw[gray,fill=gray] (2.3+0.1*\x+0.9*\i, -2.3-0.1*\y-0.9*\i) circle (0.25pt);
        }
    }
}
\foreach \i in {1, ..., 3}{
    \draw[blue,fill=blue] (2.4 + 0.9*\i,-2.4- 0.9*\i) circle (1pt);
    \foreach \j in {1, ..., 2}{
        \foreach \x in {1, ..., 3}{
            \foreach \y in {1, ..., 3}{
                \draw[blue,fill=blue] (2.1 + 0.9*\i+0.3*\j+0.1*\x,-2.1- 0.9*\i-0.3*\j-0.1*\y) circle (1pt);
            }
        }
    }
    \foreach \x in {8, ..., 9}{
        \foreach \y in {8, ..., 9}{
            \draw[blue,fill=blue] (2.3+0.1*\x+0.9*\i, -2.3-0.1*\y-0.9*\i) circle (1pt);
        }
    }
}
\foreach \x in {1, ..., 5}{
    \foreach \y in {1, ..., 5}{
        \draw[gray,fill=gray] (5.9+0.1*\x, -5.9-0.1*\y) circle (0.25pt);
    }
}
\foreach \x in {1, ..., 3}{
    \foreach \y in {1, ..., 3}{
        \draw[brown,fill=brown] (5.9+0.1*\x, -5.9-0.1*\y) circle (1pt);
    }
}
\foreach \x in {4, ..., 5}{
    \foreach \y in {4, ..., 5}{
        \draw[brown,fill=brown] (5.9+0.1*\x, -5.9-0.1*\y) circle (1pt);
    }
}
\end{tikzpicture}
\end{center}

\noindent The left-hand diagram shows all non-zero values in $U_{PQC}\mathbf{P}(\sigma)U_{PQC}^\dagger \; \forall \sigma \in S_n$ and the right-hand diagram shows the same with $\sigma$ restricted to the subgroup described above. Faint, gray dots are zero values - they are only there to allow for comparison to the blocks of the left-hand diagram. The block diagonal structure of the right-hand diagram shows that our basis is adapted to this subgroup with respect to this representation. In the left-hand diagram, black represents copies of the $S^2 = 3$ irrep, red represents copies of the $S^2 = 2$ irrep and blue and brown are $S^2 = 1$ and $S^2 = 0$ respectively. Note that the irreps of maximal spin, in this case $S^2 = 3$, are always trivial (one-dimensional) because states of maximal spin must be totally symmetric, so any $\mathbf{P}(\sigma)$ leaves them untouched. From the left-hand diagram, we can immediately discern the block diagonal structure of the representation $\mathbf{Q}$. We would find, moving down the diagonal, $1$ block of size 7, 5 blocks of size 5, 9 blocks of size 3 and 5 blocks of size $1$. Finally, we note that the basis states have been ordered in such a way as to make the above structures look nice. A different ordering does not change the invariant subspaces, but does change the visual depiction.
\end{example}

\subsection{General PQC Unitaries}\label{genPQCUnitaries}

One advantage of our algorithm for the Schur transform is that it naturally generalises to algorithms for other PQC unitaries taking us from the computational basis to other PQC bases. It is worth noting, however, that while we find this algorithm using a polynomial number of ancillas as a direct generalisation of our algorithm for the Schur transform, we do not find an algorithm for the implementation of a general PQC unitary using only a logarithmic number of ancillas, as is shown for the Schur transform in Section \ref{logManyAncillas}. This issue is discussed more in that section.

A generic PQC tree can have a much more complex structure than the sequential coupling of the PQC tree of the Schur transform. The most instructive way of presenting the algorithm for a general PQC unitary is via an example on a constant number of qubits, but we will later present the full algorithm, having introduced some new notation to describe general PQC trees.

\begin{example}
Consider the unlabelled tree corresponding to the tree from Example \ref{treeExample}, shown here for ease of reference:
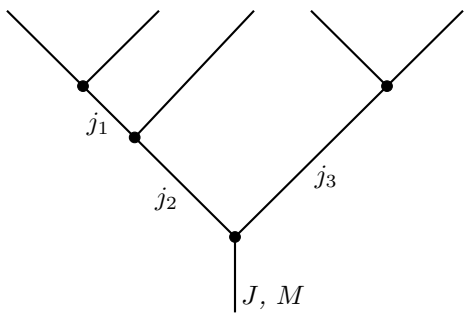
\begin{center}
\begin{tikzpicture}
\draw[fill=black] (3,1) circle (2pt);
\draw[fill=black] (1.68,2.32) circle (2pt);
\draw[fill=black] (5,3) circle (2pt);
\draw[fill=black] (1,3) circle (2pt);

\draw[thick] (3,0) -- (3,1) -- (1.68,2.32) -- (1,3) -- (0,4);
\draw[thick] (1,3) -- (2,4);
\draw[thick] (1.68,2.32) -- (3.25,4);
\draw[thick] (5,3) -- (4,4);
\draw[thick] (3,1) -- (5,3);
\draw[thick] (5,3) -- (6,4);

\node at (1.2,2.5) {$j_1$};
\node at (2.1, 1.5) {$j_2$};
\node at (4.2,1.8) {$j_3$};
\node at (3.5,0.2) {$J$, $M$};

\end{tikzpicture}
\end{center}
\noindent where we have shown the labels we will go on to use for each spin eigenvalue. We will show how to implement the PQC unitary mapping the computational basis to this basis according to our algorithm. The two-stage structure of the algorithm is very much the same, with a coupling stage coming after a pre-mapping stage. The pre-mapping stage performs 

\begin{equation}
    \ket{x} \mapsto \ket{j_1}\ket{j_2}\ket{j_3}\ket{J}\ket{M}
\end{equation}

\noindent coherently over all the computational basis states, where it is left implicit that the spin eigenvalues on the right-hand side depend on the computational basis state $\ket{x}$ from which they originate. The spin eigenvalues are encoded into computational basis states of the appropriate number of qubits, as before. The pre-mapping may be performed via

\begin{align}
    \ket{x} = \ket{x_1x_2x_3x_4x_5} &\mapsto \ket{j_1}\ket{m_1}\ket{x_3x_4x_5}\\
    &\mapsto \ket{j_1}\ket{j_2}\ket{m_2}\ket{x_4x_5}\\
    &\mapsto \ket{j_1}\ket{j_2}\ket{m_2}\ket{j_3}\ket{m_3}\\
    &\mapsto \ket{j_1}\ket{j_2}\ket{j_3}\ket{J}\ket{M}
\end{align}

\noindent where, for example, going into the last line, we act on the registers $\ket{j_2}\ket{m_2}\ket{j_3}\ket{m_3}$. The label $m_i$ is used to signify the $z$-angular momentum of the qubits with total angular momentum labelled by $j_i$. Notice that in these general PQC unitaries, we will act on four registers in general, where previously we only acted on three in general, but this is okay as the total number of qubits in question is still logarithmic. The coupling stage may then act back up the tree in the natural way:

\begin{align}
    \ket{j_1}\ket{j_2}\ket{j_3}\ket{J}\ket{M} \mapsto &\sum_{\mathclap{\substack{m_2, m_3\\}}} C^{J,M}_{j_2, m_2; j_3, m_3} \ket{j_1}\ket{j_2}\ket{m_2}\ket{j_3}\ket{m_3} \\
    \mapsto &\sum_{\mathclap{\substack{m_2, m_3\\x_4, x_5\\}}} C^{J,M}_{j_2, m_2; j_3, m_3} C^{j_3, m_3}_{1/2, x_4; 1/2, x_5}\ket{j_1}\ket{j_2}\ket{m_2}\ket{x_4x_5}\\
    \mapsto &\sum_{\mathclap{\substack{m_1, m_2, m_3\\x_3, x_4, x_5}}} C^{J,M}_{j_2, m_2; j_3, m_3} C^{j_3, m_3}_{1/2, x_4; 1/2, x_5}C^{j_2, m_2}_{j_1, m_1; 1/2, x_3}\ket{j_1}\ket{m_1}\ket{x_3x_4x_5}\\
    \mapsto &\sum_{\mathclap{\substack{m_1, m_2, m_3\\x_1, x_2, x_3, x_4, x_5}}} C^{J,M}_{j_2, m_2; j_3, m_3} C^{j_3, m_3}_{1/2, x_4; 1/2, x_5}C^{j_2, m_2}_{j_1, m_1; 1/2, x_3}C^{j_1, m_1}_{1/2, x_1; 1/2, x_2}\ket{x}\\
    =&\sum_{\mathclap{\substack{x_1, x_2, x_3, x_4, x_5}}} C^{J,M}_{j_2, x_1 + x_2 + x_3; j_3, x_4 + x_5} C^{j_3, x_4 + x_5}_{1/2, x_4; 1/2, x_5}C^{j_2, x_1 + x_2 + x_3}_{j_1, x_1 + x_2; 1/2, x_3}C^{j_1, x_1 + x_2}_{1/2, x_1; 1/2, x_2}\ket{x}.
\end{align}
\end{example}

\noindent To present the algorithm in full generality, we must introduce a way to label the vertices of our PQC trees as $v_i$ for $i \in \{1, ..., n-1\}$. $v_i$ represent subsets of the qubits: $v_i \subseteq [n] = \{1, ..., n\}$, such that the qubits labelled by $v_i$ are subsumed by the i-th vertex and we require of our numbering that if $v_i \subseteq v_j$ for $i \neq j$ then $i < j$. We can then refer to spin eigenvalues as $j_i$ and $m_i$ corresponding to those eigenvalues of the qubits subsumed by the i-th vertex. Notice that $j_{n-1} = J$ and $m_{n-1} = M$. Using this notation, it is possible to specify a PQC tree and the corresponding PQC basis using a collection of qubit subsets, $\left(v_i\right)_{i=1}^{n-1}$. The last notation we need in order to present the algorithm in full generality is two functions, $l(i)$ and $r(i)$, which specify the two qubit subsets that join to form $v_i$ i.e. $v_i = v_{l(i)} \cup v_{r(i)}$. We think of these as the `left' and `right' branches joining to form the i-th vertex. The only complication is if one of the branches joining from above to the i-th vertex represents a single qubit, and therefore does not have an i-value. There are many ways one could choose to deal with this to distinguish individual qubit labels from vertex labels, for example setting $l(i) = -3$ if the left-hand branch joining to form the i-th vertex is the third qubit. With this arrangement, we would have, for example $v_{-3} = \{3\}$, $j_{-3} = 1/2$ and $m_{-3} = x_3$. We may then present the first stage of this algorithm, the pre-mapping.

\begin{algorithm}[H]
	\caption{Pre-Mapping Stage for a General PQC Unitary\label{premapPQC}}
\textbf{Input:}  Any superposition of $n$-qubit computational basis states: $\sum_xc_x\ket{x}$ as well as some collection of qubit subsets $\left(v_i\right)_{i=1}^{n-1}$ that specify a PQC basis.\\
\textbf{Output:} The same superposition with each computational basis state mapped to the encoding on $O(n\:log(n))$ qubits of the PQC state to which it corresponds: $\sum_xc_x\ket{j_1(x)}\ket{j_2(x)}...\ket{j_{n-2}(x)}\ket{J(x)}\ket{M(x)}$.
	\begin{algorithmic}[1]
			\For{$i=1$ to $n-2$}
		    \State{Perform the mapping between computational basis states $\ket{j_{l(i)}}\ket{m_{l(i)}}\ket{j_{r(i)}}\ket{m_{r(i)}} \mapsto \ket{j_{l(i)}}\ket{j_{r(i)}}\ket{j_i}\ket{m_i}$ using some ordering.}
			      \EndFor
		\State{Perform the mapping between computational basis states $\ket{j_{l(n-1)}}\ket{m_{l(n-1)}}\ket{j_{r(n-1)}}\ket{m_{r(n-1)}} \mapsto \ket{j_{l(n-1)}}\ket{j_{r(n-1)}}\ket{J}\ket{M}$ using some ordering.}
		\end{algorithmic}
\end{algorithm}

\noindent We note that strictly speaking, in both of these algorithms, no register should be used to encode the total angular momentum on a single qubit, although this is in some sense already encapsulated in our notation, because we always use $\lceil log(k)\rceil$ qubits to encode a variable that can take at most $k$ values, and so no qubits are required to encode a variable that can take only $1$ value.

\begin{algorithm}[H]
	\caption{Coupling Stage for a General PQC Unitary\label{couplingPQC}}
\textbf{Input:} Some collection of qubit subsets $\left(v_i\right)_{i=1}^{n-1}$ that specify a PQC basis as well as any superposition of computational basis states on $O(n\;log(n))$ qubits encoding PQC states from that basis: $\sum_xc_x\ket{j_1(x)}\ket{j_2(x)}...\ket{j_{n-2}(x)}\ket{J(x)}\ket{M(x)}$.\\
\textbf{Output:} The same superposition with each computational basis state encoding mapped to its corresponding PQC state on $n$ qubits.
	\begin{algorithmic}[1]
\State{Perform the mapping on four registers: $\ket{j_{l(n-1)}}\ket{j_{r(n-1)}}\ket{J}\ket{M} \mapsto$\newline$\sum_{m_{l(n-1)}, m_{r(n-1)}}C^{J,M}_{j_{l(n-1)}, m_{l(n-1)}; j_{r(n-1)}, m_{r(n-1)}}\ket{j_{l(n-1)}}\ket{m_{l(n-1)}}\ket{j_{r(n-1)}}\ket{m_{r(n-1)}}$.}
			\For{$i=n-2$ to 1}
		    \State{Perform the mapping $\ket{j_{l(i)}}\ket{j_{r(i)}}\ket{j_i}\ket{m_i} \mapsto \sum_{m_{l(i)}, m_{r(i)}}C^{j_i, m_i}_{j_{l(i)}, m_{l(i)}; j_{r(i)}, m_{r(i)}}\ket{j_{l(i)}}\ket{m_{l(i)}}\ket{j_{r(i)}}\ket{m_{r(i)}}$.}
			      \EndFor
		\end{algorithmic}
\end{algorithm}

\noindent Appendix \ref{runtime} contains a run time analysis for the whole of the above algorithm, as well as a run time analysis for the preparation of individual PQC states.

\section{Generalised Coupling}\label{genCoupling}
We will now see how far we can push our ideas on coupling systems to define a new, very general class of quantum states and discuss the conditions under which they are computationally tractable (CT) and when they can, with our methods, be efficiently prepared on a quantum computer. Within our formalism, the computational tractability and efficient preparation results become manifest. We will show as an example how SU(N) Clebsch-Gordan coefficients fit naturally into our formalism, before looking at the same for Wigner 6-j symbols. Note that in this section we are mostly only interested in individual states, not bases as is the case with PQC trees.

\subsection{Defining Generalised Coupling}

In this general setting, each given vertex may couple any constant number of quantum systems, not just the 2 as for PQC trees. Furthermore, at each vertex, a different number of quantum systems may be coupled and different coupling coefficients may be used. The definition of coupling coefficients at one vertex is given below and after that we show how to generally couple systems using coupling coefficients.

\begin{definition}
Consider coefficients $\alpha^{J, M}_{j_1, m_1; j_2, m_2; ...; j_k, m_k} \in \mathbb{C}$ where $k \in \{2, 3, 4, ...\}$ is a constant and the indices $J, M, j_1, m_k, ..., j_k, m_k$ lie in $\mathbb{Z}$\footnote{$\mathbb{Z}$ is chosen for the sake of simplicity. $\frac{\mathbb{Z}}{2}$ is used in the case of angular momentum but any countable indexing set may be used via an injection into $\mathbb{Z}$. The $j$ and $m$ are now entirely general labels where previously they referred only to total angular momentum and $z$-angular momentum. They may be operator eigenvalues, or something else entirely.}. We call these coupling coefficients if they satisfy the following two conditions.
\\

\noindent \textbf{(GC1)} For each fixed $(J, M, j_1, j_2, ..., j_k)$, $\alpha^{J, M}_{j_1, m_1; j_2, m_2; ...; j_k, m_k}$ is non-zero only for $(m_1, ..., m_k)$ in some finite \indent \indent \; set.
\\

\noindent \textbf{(GC2)} \begin{equation}
    \sum_{m_1, ..., m_k} \alpha^{J,M}_{j_1, m_1; ...; j_k, m_k} \bar{\alpha}^{J,M'}_{j_1, m_1; ...; j_k, m_k} = \delta_{M M'} \indent \forall J, M, M', j_1, ..., j_k.
\end{equation}
\end{definition}

\noindent Notice that the usual Clebsch-Gordan coefficients satisfy this definition. For each $(J, M, j_1, j_2)$, $C^{J,M}_{j_1, m_1; j_2, m_2}$ is non-zero only if $m_1 \in \{-j_1, -j_1 + 1, ..., j_1\}$ and $m_2 \in \{-j_2, -j_2 + 1, ..., j_2\}$. They also satisfy an orthogonality relation that is more restrictive than \textbf{(GC2)}: $\sum_{m_1, m_2} C^{J,M}_{j_1, m_1; j_2, m_2} \bar{C}^{J',M'}_{j_1, m_1; j_2, m_2} = \delta_{J J'}\delta_{M M'}$ \cite{CGWiki}. Notice that having this orthogonality over the $J$-labels for Clebsch-Gordan coefficients allows Schur states and states from other PQC trees to form bases - not having this condition on general coupling coefficients is one important reason why collections of generally coupled states do not form bases without further restriction.

With these coupling coefficients, we may create generally coupled states in the natural way. Take $k$ systems for which the i-th system is considered to have some fixed $j$-value, $j_i$ and so lives in $span\{\ket{j_i,m_i} : j_i\text{ is fixed}\}$. These systems may be coupled to form a state

\begin{equation}
    \ket{J, M} = \sum_{m_1, ..., m_k} \alpha^{J, M}_{j_1, m_1; ... ; j_k, m_k}\ket{j_1, m_1} ... \ket{j_k, m_k}.\label{vertexCoupling}
\end{equation}

\noindent This allows us to build up states like the one below.

\begin{example}\label{treeExample}
Consider the following state on four qubits.
\\
\begin{vwcol}[widths={0.225,0.775}, rule = 0pt]
\phantom{M}
\begin{tikzpicture}

\draw[fill=black] (2,1.04) circle (2pt);
\draw[fill=black] (1,2.54) circle (2pt);

\node at (0.8,2.2) {$(\alpha)$};
\node at (1.7,0.74) {$(\beta)$};
\node at (1.62,2) {$j$};
\node at (2.45,0.15) {$J,M$};

\draw[thick] (2,0) -- (2,1) -- (1,2.5) -- (0,4);
\draw[thick] (1,2.5) -- (1,4);
\draw[thick] (1,2.5) -- (2,4);
\draw[thick] (2,1) -- (4,4);
\end{tikzpicture}
\\
\phantom{M}
\\
\phantom{M}
\\
\phantom{M}
\\
 = $\ket{J,M}$
\\
= $\sum_{m, x_4} \beta^{J,M}_{j,m;1,x_4}\ket{j,m}\ket{x_4}$
\\
= $\sum_{m, x_4} \beta^{J,M}_{j,m;1,x_4}\sum_{x_1, x_2, x_3}\alpha^{j,m}_{1,x_1; 1,x_2; 1,x_3}\ket{x_1x_2x_3x_4}$
\\
= $\sum_{m, x_1, ..., x_4} \beta^{J,M}_{j,m;1,x_4}\alpha^{j,m}_{1,x_1; 1,x_2; 1,x_3}\ket{x_1x_2x_3x_4}$

\end{vwcol}
\vspace{0.5cm}
\noindent where $\ket{x_i}$ are computational basis states of the individual qubits and we have used different coupling coefficients $\alpha$ and $\beta$ at each of the vertices, as shown in the diagram. We have also chosen to use a $j$-value of $1$ to refer to individual qubits in the state, but a different value could be taken.
\end{example}

\noindent The `base systems' in this example i.e. the systems at the leaves of the tree were qubits, and most of the time we would consider them to be qubits or qudits. However, the notion of coupling more general systems could be considered (but see below for a condition). States coupled with coupling coefficients are indeed quantum states as a consequence of \textbf{(GC1)} and \textbf{(GC2)}. \textbf{(GC1)} means that each sum is well-defined whereas \textbf{(GC2)} allows us to show by induction that the state has norm $1$. Indeed, in equation \eqref{vertexCoupling}, \textbf{(GC2)} gives us that $\braket{j_i,m|j_i,n} = \delta_{mn}$ for each $i = 1, ..., $k$ \implies \braket{J,N|J,M} = \delta_{MN}$ and, in particular, the state itself has norm $1$. The base case of the induction must be provided by our base systems - whatever the base systems are, their states that are coupled must be orthonormal - this is clearly true in the case of qudits. For simplicity, the rest of the section will only consider states for which the base systems are qudits.

\subsection{Computational Tractability}

We will now show that, under some slightly more restrictive conditions, these states are computationally tractable (CT), a concept defined by Van den Nest in \cite{VdNProbMethods}. As pointed out there, the notion of CTness strictly applies to families of states $\left(\ket{\psi_n}\right)_{n=1}^\infty$, where $\ket{\psi_n}$ is made up of $n$ qubits, but this fact is often omitted for brevity. A CT `state' is commonly used as shorthand for a CT state family. The notion of computational tractability provides a broad framework for studying the efficient classical simulation of quantum computations. The definition is repeated here for ease of reference.

\begin{definition}\label{CTdef}
An $n$-qubit state $\ket{\psi}$ is computationally tractable (CT) if
\begin{enumerate}
    \item One can classically sample from the distribution $\{|\braket{x|\psi}|^2 : x \in \{0,1\}^n\}$ in poly(n) time.
    \item One can classically compute $\braket{x|\psi}$ for each $x \in \{0,1\}^n$ in poly(n) time.
\end{enumerate}
\end{definition}

\noindent We will go on to show that a state formed by the coupling of $n$ qubits is computationally tractable if each of its coupling coefficients satisfy the conditions \textbf{(GC3)} - \textbf{(GC5)} and the state itself satisfies \textbf{(GC6)}, which are as follows.
\\

\noindent \textbf{(GC3)} For each $(J, M, j_1, m_1, ..., j_k, m_k)$, $\alpha^{J, M}_{j_1, m_1; j_2, m_2; ...; j_k, m_k}$ may be classically computed in poly($J$, $M$, $j_1$, \indent\indent\;\;$m_1$, ..., $j_k$, $m_k$) time.
\\

\noindent \textbf{(GC4)} For each $(J, j_1, m_1, ..., j_k, m_k)$, there is at most one $M$ for which $\alpha^{J, M}_{j_1, m_1; j_2, m_2; ...; j_k, m_k}$ is non-zero. In \indent\indent\;\;poly($J$, $j_1$, $m_1$, ..., $j_k$, $m_k$) time, one can classically compute this value of $M$, or conclude that it \indent\indent\;\;does not exist, from the input $(J, j_1, m_1, ..., j_k, m_k)$.
\\

\noindent \textbf{(GC5)} For each $(J, M, j_1, ..., j_k)$, $\alpha^{J, M}_{j_1, m_1; j_2, m_2; ...; j_k, m_k}$ is non-zero only for $(m_1, m_2, ..., m_k)$ in a set of size poly($J$, \indent\indent\;\,$M$, $j_1$, $j_2$, ..., $j_k$), which one can classically compute in poly($J$, $M$, $j_1$, $j_2$, ..., $j_k$) time.
\\

\noindent \textbf{(GC6)} For each coupling coefficient $\alpha^{J, M}_{j_1, m_1; j_2, m_2; ...; j_k, m_k}$ in the state, for each series of indices $(J, M, j_1, m_1, ..., \\ \indent\indent\;\: j_k, m_k)$ inputted to the coefficient, each index is of size poly(n) in absolute value.
\\

\noindent It is worth remarking why we have written `in poly($J$, $M$, $j_1$, $m_1$, ..., $j_k$, $m_k$) time' (etc.) in \textbf{(GC3)} - \textbf{(GC5)} and then said that each of these inputs must be of size poly(n) in \textbf{(GC6)} rather than just writing `in poly(n) time' in \textbf{(GC3)} - \textbf{(GC5)} and omitting \textbf{(GC6)}. The reason is simply to be most faithful to the prototypical example of SU(2) Clebsch-Gordan coefficients and PQC states. The Clebsch-Gordan coefficient $C^{J,M}_{j_1, m_1; j_2, m_2}$ can be computed in poly($J$, $M$, $j_1$, $m_1$, $j_2$, $m_2$) time via its relation to the corresponding Wigner 3j-coefficient $C^{J,M}_{j_1, m_1; j_2, m_2} = (-1)^{-j_1 + j_2 - M}\sqrt{2J+1}\begin{pmatrix} j_1 & j_2 & J \\ m_1 & m_2 & -M \end{pmatrix}$ \cite{CGWiki} and the Racah formula for Wigner 3j-coefficients \cite{Wigner3jWolfram}. Additionally, each $j$ and $m$-value in a PQC state inputted into a Clebsch-Gordan coefficient is bounded between $-\frac{n}{2}$ and $\frac{n}{2}$, so PQC states satisfy \textbf{(GC6)}.

Clebsch-Gordan coefficients also satisfy \textbf{(GC4)} and \textbf{(GC5)}, since $C^{J,M}_{j_1, m_1; j_2, m_2}$ is zero unless $M = m_1 + m_2$ and $(m_1, m_2) \in \{-j_1, -j_1 + 1, ..., j_1\} \times \{-j_2, -j_2 + 1, ..., j_2\}$. From this, we see that \textbf{(GC4)} is a generalisation of the conservation of angular momentum, which is not needed to define a quantum state but is necessary for computational tractability (at least by these methods).

To show computational tractability, we will use the same ideas as those employed by Havlíček in \cite{HavlicekThesis} to show the computational tractability of PQC states, although some care will be required to make sure that the proof goes through in this general a setting. It will make most sense to start with the second of the two points in Definition \ref{CTdef}. Consider what happens when we expand such a coupled state in the same way as we did in Example \ref{treeExample}: from the root upwards. From the general expression of Equation \eqref{vertexCoupling}, suppose we further expand the first subsystem, assuming it can be further expanded. This gives

\begin{align}
    \ket{J, M} &= \sum_{m_1, ..., m_k} \alpha^{J, M}_{j_1, m_1; ... ; j_k, m_k}\left(\sum_{n_1, ..., n_l}\beta^{j_1, m_1}_{i_1, n_1; ...; i_k, n_k}\ket{i_1, n_1} ... \ket{i_l, n_l}\right)\ket{j_2, m_2} ... \ket{j_k, m_k}\\
    &= \sum_{m_2, ..., m_k}\sum_{n_1, ..., n_l} \alpha^{J, M}_{j_1, \hat{m}_1; ... ; j_k, m_k}\beta^{j_1, \hat{m}_1}_{i_1, n_1; ...; i_l, n_l}\left(\vphantom{\int}\ket{i_1, n_1} ... \ket{i_l, n_l}\right)\ket{j_2, m_2} ... \ket{j_k, m_k}\label{firstSubsystemExpansion}
\end{align}

\noindent where a hat denotes the unique value for which the coefficient $\beta$ is non-zero as given by \textbf{(GC4)}. We can see from this that expanding all $k$ subsystems (assuming they can all be expanded) leads to the sum over the $m_i$ vanishing. Thus, in general, $\ket{J,M}$ may be expressed as a linear combination of computational basis states where each coefficient is a product of at most $(n-1)$ coupling coefficients (because we have one coupling coefficient for each vertex). Moreover, we can classically and efficiently calculate this product, now by moving from the leaves of the tree to the root. To calculate $\braket{x|\psi}$, we work down the tree, where at each vertex, except the root, we compute the value of $M$ given by \textbf{(GC4)} that gives a non-zero coupling coefficient, returning the answer $\braket{x|\psi} = 0$ if there is none. At the root, the $m$-value in the upper line of the coupling coefficient is fixed to M. We can then calculate the coefficients and the product using \textbf{(GC3)}. All of the above is made possible by \textbf{(GC6)}.

\begin{example}
For the state in Example \ref{treeExample}, we have 
\begin{equation}
\braket{x|\psi} = \begin{cases}\beta^{J,M}_{j,\hat{m};1,x_4}\alpha^{j,\hat{m}}_{1,x_1;1,x_2;1,x_3} & \text{if }\exists \: m \; s.t. \; \alpha^{j,m}_{1,x_1;1,x_2;1,x_3} \neq 0\\
0 & \text{otherwise.}\end{cases}
\end{equation}

\noindent Our algorithm starts by calculating if there is an $m$ for which $\alpha^{j,m}_{1,x_1;1,x_2;1,x_3} \neq 0$, returning $\braket{x|\psi} = 0$ if not. If there is, it calls this value $\hat{m}$ and returns the value $\beta^{J,M}_{j,\hat{m};1,x_4}\alpha^{j,\hat{m}}_{1,x_1;1,x_2;1,x_3}$. It is instructive to note that if there is no $m$ such that $\beta^{J,m}_{j,\hat{m};1,x_4} \neq 0$, or if such an $m$ does exist but it is not equal to our fixed $M$, $\braket{x|\psi} = 0$ and indeed our algorithm returns $0$.
\end{example}

\noindent We now turn to the first part of computational tractability, sampling from $\left\{|\braket{x|\psi}|^2\right\}$. For this algorithm, we will work from the root to the leaves, performing one sample for each vertex. At each vertex, we sample the $m$-values of the subsystems above the vertex, fixing these values to those found by the sample, before proceeding further up the tree. Thus, suppose at the root we have the expansion given by Equation \eqref{vertexCoupling}, repeated here for ease of reference:

\begin{equation}
    \ket{J, M} = \sum_{m_1, ..., m_k} \alpha^{J, M}_{j_1, m_1; ... ; j_k, m_k}\ket{j_1, m_1} ... \ket{j_k, m_k}.\label{vertexCoupling2}
\end{equation}

\noindent The first step will be to sample from the distribution $\left\{\left|\alpha^{J, M}_{j_1, m_1; ... ; j_k, m_k}\right|^2\right\}_{(m_1, ..., m_k)}$ which we can do efficiently as a consequence of \textbf{(GC3)}, \textbf{(GC5)} and \textbf{(GC6)} and then fix the $m$-values for these subsystems to the values found in this sample. Then, suppose that the first subsystem can be further decomposed in the same way as it was in Equation \eqref{firstSubsystemExpansion}. In this case, we will next sample from the distribution $\left\{\left|\beta^{j_1, m_1}_{i_1, n_1; ... ; i_l, n_l}\right|^2\right\}_{(n_1, ..., n_l)}$, and so on. In total, we make as many of these samples as there are vertices in the tree, which is at most $(n-1)$, as already discussed, and so the whole procedure takes place in poly(n) time. Care will be necessary in showing that this algorithm samples from the correct distribution. To show this, we start from the general expression

\begin{equation}
    \braket{x|\psi} = \alpha^{J^{(root)}, M^{(root)}, root}_{j_1^{(root)}, \hat{m}_1^{root}; ...; j_{k_{root}}^{(root)},\hat{m}_{k_{root}}^{(root)}} \prod_{v \in \{Vertices\} \setminus \{root\}} \alpha^{J^{(v)}, \hat{M}^{(v)}, v}_{j_1^{(v)}, \hat{m}_1^{v}; ...; j_{k_{v}}^{(v)},\hat{m}_{k_{v}}^{(v)}}.\label{generalOverlap}
\end{equation}

\noindent In this expression, we have added a third entry to the upper line for each coupling coefficient, which only serves as an easy way of referring to different coupling coefficients at each vertex. At a vertex $v$, $k_v$ systems are coupled; the i-th system being coupled has $j$ and $m$-values $j_i^{(v)}$ and $m_i^{(v)}$ respectively while the resulting system (below the vertex) has $j$ and $m$-values $J^{(v)}$ and $M^{(v)}$ respectively. As above, we use a hat to denote an $m$-value that is computed using \textbf{(GC4)}, although here, for convenience, if such an $m$-value does not exist, we set $\hat{m}$ to any value, at which point the coefficient will return $0$, causing the overlap $\braket{x|\psi}$ to vanish also, as it should. Also for convenience, if the i-th system coupled at vertex $v$ is the j-th qubit, we let $\hat{m}_i^v = x_j$.

By repeatedly applying \textbf{(GC4)}, we see that in order to sample $(x_1, ..., x_n)$, there is at most one sequence of $m$-values that must be sampled in order to make it possible to sample this bit string, and if this sequence exists, it is exactly the hatted values given in Equation \eqref{generalOverlap}. The probability that the bit string $x$ is sampled is exactly the probability that each correct $m$-value is sampled and then each correct bit $x_i$ is sampled, which is exactly $|\braket{x|\psi}|^2$. This will be significantly elucidated by an example.

\begin{example}
Returning once again to the state in Example \ref{treeExample}, we have $|\braket{x|\psi}|^2 = \left|\beta^{J,M}_{j,\hat{m};1,x_4}\right|^2\left|\alpha^{j,\hat{m}}_{1,x_1;1,x_2;1,x_3}\right|^2$, where now we assign $\hat{m}$ any value in the case that $\nexists$ $m$ $s.t.$ $\alpha^{j,m}_{1,x_1;1,x_2;1,x_3} \neq 0$. Our algorithm samples from the distribution $\left\{\left|\beta^{J,M}_{j,m;1,x_4}\right|^2\right\}_{(m,x_4)}$, fixes the sampled values, and then samples from the distribution $\left\{\left|\alpha^{j,m}_{1,x_1;1,x_2;1,x_3}\right|^2\right\}_{(x_1, x_2, x_3)}$. Let us ask what the probability is that we sample a given bit string $y = (y_1, y_2, y_3, y_4)$. If there is no $m$ that gives a non-zero $\alpha^{j,m}_{1,y_1;1,y_2;1,y_3}$, then it will be impossible to sample this bit string and, indeed, $\braket{y|\psi} = 0$. If there is such an $m$, call it $\hat{m}$, it will only be possible to sample the bit string $y$ if in the first step we sample $(m, x_4) = (\hat{m}, y_4)$, which occurs with probability $\left|\beta^{J,M}_{j,\hat{m};1,y_4}\right|^2$. The rest of the bit string $y$ is then sampled with probability $\left|\alpha^{j,\hat{m}}_{1,y_1;1,y_2;1,y_3}\right|^2$. Thus, in total, we can see that the probability we sample the bit string $y$ is exactly $|\braket{y|\psi}|^2$.
\end{example}

\noindent It is interesting to consider the computational tractability of this broad family of states within the wider context of the theory of classical simulation. \cite{VdNProbMethods} identifies certain families of quantum states known already to be CT. Most notably for this discussion, these include stabiliser states (\hspace{1sp}\cite{gottesman1998heisenberg}, \cite{dehaene2003clifford}) and states formed from the application of a polynomial number of nearest neighbour matchgates to a computational basis state (\hspace{1sp}\cite{valiant2001quantum}). In both of these important cases, there is a closed set of gates that can be applied to keep the state CT (namely Cliffords in the former case). It is interesting to consider that the computational tractability of these generally coupled states does not come from structure endowed by any gate set, but from the structure of the states themselves. For example, we know that Schur states are CT, but if we apply the Schur transform to a Schur state, it is not straightforward to see how the result is CT, and indeed, we do not know. Section \ref{CTgates} contains a discussion on gates that preserve CTness. We illuminate the difficulties involved in defining a notion of CT-preservation, showing that if one is naive, it appears that any unitary gate is arbitrarily close to CT-preserving.

\subsection{Efficient State Preparation}\label{statePrep}

We will now develop algorithms in the style of those of Section \ref{PQCUnitaries} that allow us to efficiently prepare generally coupled states under the restriction of \textbf{(GC6)} and whose coupling coefficients are all under the restrictions \textbf{(GC3)} and \textbf{(GC5)}. Note that \textbf{(GC4)} is not required for efficient state preparation, but since it is sufficient for computational tractability, \textbf{(GC4)} demarcates the regime of potentially interesting quantum algorithms. To construct this algorithm, the idea is very much the same as that of the coupling stage of the general PQC unitaries, Algorithm \ref{couplingPQC}, except that each vertex corresponds to a unitary rotation of any constant number of systems, not just 2, as well as the difference that here we are only interested in preparing individual states rather than rotating the computational basis into a whole other basis, as is the case with PQC unitaries. This latter fact leads to some simplifications, as we now see. In particular, pre-mappping is not required.

Suppose we wish to prepare a generally coupled state on qudits with internal $j$-values $j_1, ..., j_a$ and root $J$ and $m$-values $J$ and $M$ respectively (and some given tree structure, as well as some given coupling coefficients). We may simply begin the preparation with the state $\ket{j_1}\ket{j_2}...\ket{j_a}\ket{J}\ket{M}$. This looks much like our state in our procedure for implementing a PQC unitary after the pre-mapping, except here, we do not need to encode each $j$ and $m$-value into a logarithmic number of qudits. We are only preparing one state and so these values can only be one thing. As such, we may `encode' these values into $1$ qudit - in any state. For definiteness, we could say that the starting state is, in fact, $\ket{0^{a+2}}$, and the labels serve only as names to call the qudits.

At the root, suppose that the tree splits according to the general expression of Equation \eqref{vertexCoupling2}. Our first unitary rotation will act on the registers $\ket{J}$, $\ket{M}$, as well as all the registers $\ket{j_i}$ such that $j_i$ refers to an internal $j$-value immediately above the root\footnote{This is unless one of the branches above the root is an individual qudit, in which case there is no $j$-register to act on - see the first rotation of Example \ref{statePrepExample}.}. Suppose  that these are $(j_1, ..., j_k)$ without loss of generality. Our first rotation then simply enacts

\begin{equation}
    \ket{j_1}...\ket{j_a}\ket{J}\ket{M} \mapsto \sum_{m_1, ..., m_k}\alpha^{J, M}_{j_1, m_1; ... ; j_k, m_k}\ket{j_1}\ket{m_1}...\ket{j_k}\ket{m_k}\ket{j_{k+1}}...\ket{j_a}.
\end{equation}

\noindent The individual kets on the right-hand side are, as always, computational basis states encoding the given values. Note, however, that the $m_i$ can vary over many values, unlike the $j_i$ and so in general these values must be encoded into more than one qudit. Because of \textbf{(GC5)} and \textbf{(GC6)}, each $\ket{m_i}$ register may consist of logarithmically many qudits, and so in total this transformation acts on logarithmically many qudits and has classically, efficiently computable matrix elements (due to \textbf{(GC3)} and \textbf{(GC6)}), and so may be efficiently implemented. Much like in the procedure for the implementation of PQC unitaries, we apply one such unitary for each vertex in the tree, proceeding from the root to the leaves, rendering the whole process efficient.

\begin{example}\label{statePrepExample}

We can now show how to prepare the state from Example \ref{treeExample}, under the assumptions that both its coupling coefficients satisfy \textbf{(GC3)} and \textbf{(GC5)} (but need not satisfy \textbf{(GC4)}) and the state itself satisfies \textbf{(GC6)}. We start from $\ket{j}\ket{J}\ket{M}$, where, as discussed above, each register may be as small as one qubit. Our procedure performs the following:

\begin{align}
    \ket{j}\ket{J}\ket{M} &\mapsto \sum_{m, x_4} \beta^{J,M}_{j,m;1,x_4}\ket{j}\ket{m}\ket{x_4}\\
    &\mapsto \sum_{m, x_4}\beta^{J,M}_{j,m;1,x_4}\sum_{x_1, x_2, x_3}\alpha^{j,m}_{1,x_1;1,x_2;1,x_3}\ket{x_1x_2x_3x_4}
\end{align}

\noindent As discussed above, note that the only register here needing more than one qubit in general is $\ket{m}$.

\end{example}

\noindent It is interesting to note that this procedure automatically uses a logarithmic number of ancillas in the case that the tree has the same structure as the Schur Transform (sequential coupling). This is because there will only ever be at most one register containing more than one qubit, which is some $\ket{m}$. This is interesting because we find it to be possible to reduce the number of ancillary qubits for the Schur Transform to logarithmic (see Section \ref{logManyAncillas}) but not for general PQC unitaries, although the reasons in each case are somewhat different.

\subsection{SU(N) Clebsch-Gordan Coefficients}\label{SUN}

We will now show that SU(N) Clebsch-Gordan coefficients fit very naturally into our formalism, finding that the computational tractability axiom \textbf{(GC4)} cannot be guaranteed for $N > 2$. In fact, we will show explicitly that it is violated for SU(3) and it seems very likely that it is similarly violated for all SU(N) with $N > 2$. We will review very briefly the facts on these coefficients that we find useful, but \cite{alex2011numerical} provides a very clear and complete source of further information.

An irrep of SU(N) may be specified by a list of $N$ integers $S = (m_{1,N}, ..., m_{N,N})$ for which $m_{k,N} \geq m_{k+1, N} \; \forall k$. $\left(m_{k,N}\right)_{k=1}^N$ and $\left(m_{k,N} + c\right)_{k=1}^N$ specify the same irrep for any $c \in \mathbb{Z}$. Up to this equivalence, irreps are uniquely specified, so $N-1$ independent components specify an SU(N) irrep. To avoid ambiguity, it is common to set $m_{N,N} = 0$. Within the irrep $S = \left(m_{k, N}\right)_{k=1}^N$, the states may be labelled by triangular arrays of integers $M = \left(m_{k, l}\right)$ called GT-patterns:

\begin{equation}
M = \begin{pmatrix} m_{1,N} \hspace{0.5cm} m_{2,N} \hspace{1cm} \ldots \hspace{0.8cm} m_{N,N} \\
m_{1,N-1} \hspace{0.5cm} \ldots \hspace{0.5cm} m_{N-1, N-1}\\
\ddots \hspace{1cm} \reflectbox{\(\ddots\)}\\
m_{1,1}
\end{pmatrix}
\end{equation}

\noindent satisfying $m_{k,l} \geq m_{k,l-1} \geq m_{k+1,l}$, which is known as the betweenness condition. The dimension of the irrep $S = \left(m_{k, N}\right)$ can then be found by counting the number of such triangular arrays allowed by a fixed top row $\left(m_{k, N}\right)$, and there exists a closed formula for this. 

The higher dimensional analogue of the j-variable for SU(2), the irrep label, is therefore $\left(m_{k, N}\right)$. What is the analogue of the m-variable? In SU(2), $m$ is an eigenvalue of the operator $J_3$. For SU(N), each of the states $\ket{M}$ is a simultaneous eigenstate of the $N-1$ operators $J_z^{(l)}$, for $1 \leq l \leq N-1$:

\begin{equation}
    J_z^{(l)}\ket{M} = \lambda_l^M\ket{M}
\end{equation}

\noindent where $\left(\lambda_l^M\right)_{l=1}^{N-1}$ is known as the $z$-weight of the state $\ket{M}$ (we will often refer to this simply as the `weight' of $\ket{M}$). For completeness, we note that $\lambda_l^M = \sigma_l^M - \frac{1}{2}\left(\sigma_{l+1}^M + \sigma_{l-1}^M\right)$, where $\sigma_l^M = \sum_{k=1}^lm_{k,l}$ for $l \geq 1$ and $\sigma_0^M = 0$. It is here that we encounter the difference between SU(2) and SU(N) for $N > 2$ that will be most important for us. For SU(2), the weights within each irrep are non-degenerate. Indeed, every state within every irrep of SU(2) can be labelled by some

\begin{equation}
    M = \begin{pmatrix} 2j \hspace{1.2cm} 0 \\ j-m \end{pmatrix}
\end{equation}

\noindent which is better known simply as $\ket{j,m}$. This is the unique state in the $j$-th irrep with the $z$-weight $-m$. However, for SU(N) for $N > 2$, the weights may have some multiplicity. This is referred to as `inner multiplicity', and it will have important consequences for us. Thus, in terms of eigenvalues, the analogue of the m-variable for general SU(N) is the weight $\left(\lambda_l^M\right)$, but in terms of labelling states, it is the lower $N-1$ rows of the GT-pattern, but we will refer to the whole GT-pattern $M$ for this task for simplicity.

When the tensor product of two SU(N) irreps is taken, the result is isomorphic to a direct sum of SU(N) irreps. In terms of the representation spaces, this may be expressed as 

\begin{equation}
    V^S \otimes V^{S'} = \oplus_{S''}N^{S''}_{S\;S'} V^{S''}\label{decomp}
\end{equation}

\noindent where $V^S$ is a representation space of the irrep S. The positive integers $N^{S''}_{S\;S'}$ denote the fact that there are, in general, multiple copies of the same irrep on the right-hand side. This is another notable difference between SU(2) and SU(N) for $N > 2$ and is known as `outer multiplicity'; for SU(2), because $j \otimes j' = |j-j'| \oplus ... \oplus (j+j')$, $N^{S''}_{S\;S'} \leq 1$. Outer multiplicity will not turn out to have significant consequences for us, however.

The equation \eqref{decomp} tells us that we may write 

\begin{equation}
    \ket{M'', \alpha} = \sum_{M, M'} C^{M'', \alpha}_{M, M'} \ket{M}\otimes\ket{M'} \label{unitary}
\end{equation}

\noindent which is a change of basis between the tensor product basis $\left\{\ket{M}\otimes\ket{M'}\right\}$ and the basis displaying the decomposition $\left\{\ket{M'', \alpha}\right\}$. The coefficients of the change of basis, $C^{M'', \alpha}_{M, M'}$, are the Clebsch-Gordan coefficients for SU(N). $\alpha = 1, ..., N^{S''}_{S\;S'}$ indexes the outer multiplicity of the irreps $S''$ on the right-hand side of equation \eqref{decomp}. The natural way to package these labels into our formalism is to write coefficients $\alpha^{\tilde{J},\tilde{M}}_{\tilde{j_1}, \tilde{m_1}; \tilde{j_2}, \tilde{m_2}}$\footnote{Tildes are used to differentiate generic $m$-labels such as $\tilde{M}$ from GT-patterns such as M.}, where $\tilde{J} = (S'', \alpha)$, $\tilde{M} = M''$, $\tilde{j_1} = S$, $\tilde{m_1} = M$, $\tilde{j_2} = S'$ and $\tilde{m_2} = M'$\footnote{Note that a single irrep is referred to twice in the coefficient - for example the irrep $S$ is encoded both in $S$ and in $M$ - but, again, it is best to do this for simplicity.}.

We will now show that of the conditions \textbf{(GC1 - 5)}, all are true for constant $N$ (i.e. $N = O(1)$) for these coefficients, other than \textbf{(GC4)}. \textbf{(GC2)} is true since the Clebsch-Gordan coefficients form a unitary in equation \eqref{unitary}. Indeed, as given in \cite{alex2011numerical}, 

\begin{equation}
    \sum_{M,M'}C^{M'', \alpha}_{M, M'}\left(C^{\hat{M}'', \hat{\alpha}}_{M, M'}\right)^* = \delta_{M'', \hat{M}''}\delta_{\alpha, \hat{\alpha}}.
\end{equation}

\noindent \textbf{(GC3)} is then true as shown in the same paper. \textbf{(GC5)} can be shown for $N = O(1)$, immediately implying \textbf{(GC1)} by a simple argument. We note that the simplest way to talk about the `size' of an irrep label $S$ or a state $M$ is to consider only the integer $m_{1,N}$, as this bounds the other integers in the multi-indices from above. Then, considering some $(S'', \alpha)$, $M''$, $S$ and $S'$, $(M, M')$ may run over a set of size at most $dim(S)\cdot dim(S')$. Using the expression for the dimension of an irrep $S$ from \cite{alex2011numerical}, we get 

\begin{equation}
    dim(S) = \prod_{1 \leq k < k' \leq N} \left(1 + \frac{m_{k,N} - m_{k',N}}{k'-k}\right) \leq \prod_{1 \leq k < k' \leq N} \left(1 + m_{1,N}\right) = poly(m_{1,N})
\end{equation}

\noindent for $N = O(1)$\footnote{In fact, \textbf{(GC1)} and \textbf{(GC2)} are both satisfied for any $N$.}. We can thus efficiently prepare coupling using these coefficients for constant $N$, as long as the state satisfies \textbf{(GC6)}. Note that \textbf{(GC6)} is satisfied naturally in the case of the Schur transform, or the analogues of PQC unitaries, for SU(N). This is true because when the product of two SU(N) irreps $S$ and $S'$ is taken, the irrep $S''$ in the decomposition with the highest $m_{1,N}$ is simply $S'' = S + S'$. \textbf{(GC6)} is then seen to be true, since we take the base systems to be qudits of dimension d = $N$, which are in the fundamental representation $S = (1, 0, ..., 0)$. These higher dimensional PQC unitaries may be implemented in very much the same way as the algorithms of Section \ref{PQCUnitaries} implement them for SU(2), as already mentioned.

What is most interesting for us is that the SU(N) Clebsch-Gordan coefficients do not respect \textbf{(GC4)}, and so SU(N) Schur states and PQC states do not appear to be computationally tractable, at least in the same way as the usual PQC states. Indeed, while it is true that if the $z$-weight of the state $\ket{M''}$ does not equal the sum of the $z$-weights of $\ket{M}$ and $\ket{M'}$ then $C^{M'', \alpha}_{M, M'} = 0$, which is a generalisation of the usual conservation of angular momentum for SU(2), the inner multiplicities of states in the irreps of the direct sum decomposition mean that \textbf{(GC4)} is not respected. Indeed, we can show an explicit violation for SU(3). Here, it is true that

\begin{equation}
    (1,0,0) \otimes (1,1,0) = (2,1,0) \oplus (0,0,0)
\end{equation}

\noindent which may be more familiar to physicists when written as $\mathbf{3} \otimes \mathbf{\bar{3}} = \mathbf{8} \oplus \mathbf{1}$. The eight dimensional irrep $(2,1,0)$ has two states with $z$-weight $(0,0)$, which may be written as

\begin{equation}
    \frac{\ket{\mathbf{3},1}\otimes\ket{\mathbf{\bar{3}},3} + \ket{\mathbf{3}, 2}\otimes\ket{\mathbf{\bar{3}},2}}{\sqrt{2}}
\end{equation}

\noindent and

\begin{equation}
    -\frac{1}{\sqrt{6}}\ket{\mathbf{3},1}\otimes\ket{\mathbf{\bar{3}},3} + \frac{1}{\sqrt{6}}\ket{\mathbf{3},2}\otimes\ket{\mathbf{\bar{3}},2} + \sqrt{\frac{2}{3}}\ket{\mathbf{3},3}\otimes\ket{\mathbf{\bar{3}}, 1}
\end{equation}

\noindent where $\ket{\mathbf{3}, i}$ and $\ket{\mathbf{\bar{3}}, j}$ label the states of $(1,0,0)$ and $(1,1,0)$ respectively, giving the violation of \textbf{(GC4)}. The fact that these states are efficiently preparable but do not immediately appear to be computationally tractable, as well as the fact that the unitaries may be efficiently implemented, opens up a possible new avenue for interesting quantum algorithms.

\subsection{Wigner 6-j Symbols}

The Wigner 6-j symbols $\begin{Bmatrix} a & b & f \\ c & e & d \end{Bmatrix}$ are important objects in the recoupling theory of SU(2) \cite{martin2019primer} but arise in many other areas \cite{yutsis1962mathematical}, \cite{varshalovich1988quantum}, \cite{haggard2011asymptotic}, \cite{de2003orthogonal}. For us, they are most naturally exposed via their relation to the recoupling tensor \cite{jordan2009permutational}

\begin{equation}
    \begin{bmatrix}a & b & f \\ c & e & d \end{bmatrix} = (-1)^{a+b+c+e}\sqrt{(2d+1)(2f+1)}\begin{Bmatrix}a & b & f \\ c & e & d \end{Bmatrix}
\end{equation}

\noindent which forms the coefficients of the following change of basis:

\begin{center}
\begin{tikzpicture}

\draw[fill=black] (-3,0.52) circle (2pt);
\draw[fill=black] (-2,1.5) circle (2pt);

\draw[fill=black] (3,0.52) circle (2pt);
\draw[fill=black] (2,1.5) circle (2pt);

\draw[thick] (-3,-0.25) -- (-3,0.5) -- (-5,2.5);
\draw[thick] (-3,0.5) -- (-1,2.5);
\draw[thick] (-2,1.5) -- (-3,2.5);

\draw[thick] (3,-0.25) -- (3,0.5) -- (1,2.5);
\draw[thick] (3,0.5) -- (5,2.5);
\draw[thick] (2,1.5) -- (3,2.5);

\node at (-2.77,0) {$e$};
\node at (-5,2.7) {$a$};
\node at (-3,2.7) {$b$};
\node at (-1,2.7) {$c$};
\node at (-2.35,0.85) {$d$};

\node at (3.23,0) {$e$};
\node at (1,2.7) {$a$};
\node at (3,2.7) {$b$};
\node at (5,2.7) {$c$};
\node at (2.31,0.83) {$f$};
\node at (-0.15,1.1) {$=\underset{f}{\bigsigma}\begin{bmatrix}a & b & f \\ c & e & d\end{bmatrix}$};
\end{tikzpicture}
\end{center}

\noindent This uses the PQC notation as introduced in Section \ref{PQCUnitaries}, where all labels are $j$-variables and the change in tree structure may occur at any node in a given tree. From this it is natural, and correct, to assume that there is a relation between Wigner 6-j symbols and SU(2) Clebsch-Gordan coefficients. By virtue of being the components of a change of basis, these satisfy the orthogonality relation \cite{6jwiki}

\begin{equation}
\sum_f \begin{bmatrix} a & b & f \\ c & e & d \end{bmatrix} \begin{bmatrix} a & b & f \\ c & e & d' \end{bmatrix} = \delta_{d \; d'} \{e \;\; a \;\; d\}\{d \;\; b \;\; c\}\label{6jorthog}
\end{equation}

\noindent where $\{\alpha\;\; \beta \;\; \gamma\}$ is the `triangular delta', equalling

\begin{equation}
\begin{cases} $1$ & \text{if } \alpha \in \left\{\left|\beta-\gamma\right|, ..., \beta + \gamma\right\}\\
0 & \text{otherwise.}\end{cases}
\end{equation}

\noindent Wigner 6-j symbols may be packaged into our formalism as

\begin{equation}
\alpha^{J, M}_{j_1, m_1; j_2, m_2} = \delta_{m_1  m_2} \begin{bmatrix}J & b & m_1 \\ j_1 & j_2 & M \end{bmatrix} = \delta_{m_1 m_2}(-1)^{J + b + j_1 + j_2}\sqrt{(2M+1)(2m_1+1)}\begin{Bmatrix}J & b & m_1 \\ j_1 & j_2 & M\end{Bmatrix}\label{package}
\end{equation}

\noindent where b is any $j$-variable that we may choose (it may be different in each coefficient). The only constraint that we impose is that the triangular deltas of equation \eqref{6jorthog} are always equal to $1$, which is a very natural condition to impose as these are the `triangle conditions' given to us by the usual rules of addition of angular momentum. We may now examine our axioms with this packaging.

\textbf{(GC2)} is quickly seen to be satisfied, noting that Wigner 6-j symbols are real:

\begin{align}
\sum_{m_1, m_2} \alpha^{J, M}_{j_1, m_1; j_2, m_2}\bar{\alpha}^{J, M}_{j_1, m_1; ; j_2, m_2} &= \label{doubleSum}
\sum_{m_1} \begin{bmatrix}J & b & m_1 \\ j_1 & j_2 & M \end{bmatrix}\begin{bmatrix}J & b & m_1 \\ j_1 & j_2 & M' \end{bmatrix}\\
&=\delta_{M M'}
\end{align}

\noindent under the above stipulation. At this point, the reader may consider the packaging of the variables $m_1$ and $m_2$ into the $\delta_{m_1 m_2}$ of equation \eqref{package} to be somewhat unnatural or perhaps something of a `cheat' to make the double sum of equation \eqref{doubleSum} collapse to a single sum so that we may employ the orthogonality relation. However, this trick may be considered in very much the same spirit as the same double sum in the case of SU(2) Clebsch-Gordan coefficients secretly being a single sum due to the conservation of angular momentum relation.

\textbf{(GC3)} is satisfied as seen via the formulae of \cite{jordan2009permutational}. \textbf{(GC5)} is seen to be true via the fact that $\begin{bmatrix}J & b & m_1 \\ j_1 & j_2 & M \end{bmatrix}$ is zero unless $\{m_1\;\; j_1\;\; j_2\} = 1$, as discussed in \cite{6jwiki}, which also gives us \textbf{(GC1)}. We may then again show the explicit violation of \textbf{(GC4)}, opening up a further set of possible interesting quantum algorithms. Indeed, both $\begin{bmatrix}1/2 & 1/2 & 0 \\ 1/2 & 1/2 & 0 \end{bmatrix}$ and $\begin{bmatrix}1/2 & 1/2 & 0 \\ 1/2 & 1/2 & 1 \end{bmatrix}$ are non-zero.

\section{Reducing to Logarithmically Many Ancillas}\label{logManyAncillas}

Here, we show how to reduce the number of ancillary qubits used by our algorithm for the Schur transform to logarithmic, thereby matching the smallest number of ancillas used in previous implementations. Furthermore, this algorithm has the same run time (asymptotically) as our original one, as is discussed in Appendix \ref{runtime}, which is the lowest proved gate sequence length for the quantum Schur transform of which we are aware. Interestingly, we will not be able to perform this same reduction to a logarithmic number of ancillas on general PQC unitaries - the problem as to whether this can be done is left open. 

Our reduction to logarithmically many ancillas for the Schur transform relies on the following observation. As well as the root $J$ and $M$ values, a Schur state is specified by $n-2$ internal $j$-values, $j_1, ..., j_{n-2}$. Given $j_p$ for $p < n-2$, $j_{p+1}$ may either be $j_p+1/2$ or $j_p-1/2$, unless $j_p = 0$, in which case $j_{p+1} = 1/2$, by the usual rules of addition of angular momentum. Therefore, $j_{p+1}$ has at most two `choices' when given the value of $j_p$, which is a decision that we may encode into a qubit\footnote{Note that for a general PQC unitary, given the two $j$-values above a vertex, the $j$-value below the vertex may take more than 2 values - in general it may take linearly many values. This prevents our reduction to a logarithmic number of ancillary qubits from generalising to all PQC unitaries.}.

The natural way to do this is using Yamanouchi symbols - explained, for example, in \cite{SchurSampling}. These are bit strings representing such sequences of $j$-values. The idea is, quite simply, that a $1$ represents an increased $j$-value, a $0$ represents a decrease, and we start from $1/2$. Thus, the sequence of $j$-values $j_1 = 0$, $j_2 = 1/2$, $j_3 = 1$ is represented by the Yamanouchi symbol 011.

The first stage of the algorithm was the pre-mapping which coherently performs

\begin{equation}
    \ket{x_1 ... x_n} \mapsto \ket{j_1}\ket{j_2}...\ket{j_{n-2}}\ket{J}\ket{M}
\end{equation}

\noindent where recall the individual kets on the right-hand side are computational basis states serving simply as labels. We will replace this with the pre-mapping 

\begin{equation}
    \ket{x_1 ... x_n} \mapsto \ket{y_1 ... y_{n-2}}\ket{j_{n-2}}\ket{J}\ket{M}
\end{equation}

\noindent where $y_1 ... y_{n-2}$ is the Yamanouchi symbol corresponding to $j_1, ..., j_{n-2}$ and so $\ket{y_1 ... y_{n-2}}$ represents $n-2$ qubits in a computational basis state. The pre-mapping may be performed as follows, using only a logarithmic number of ancillas. All of the following operations may be efficiently implemented as they will all act on a logarithmic number of qubits and have classically, efficiently computable matrix elements, although in some cases we will have to take care to make sure the operation is in fact unitary. As in previous explanations, it will not be explicitly mentioned when register sizes change but we again make the comment that when register sizes get smaller, ancillary qubits that are no longer needed are returned to the state $\ket{0}$, therefore becoming un-entangled with the bulk of the state.

Begin with the first three steps of the usual pre-mapping $\ket{x_1 ... x_n} \mapsto \ket{j_1}\ket{j_2}\ket{j_3}\ket{m_3}\ket{x_5 ... x_n}$. We will then make a copy of both $\ket{j_1}$ and $\ket{j_2}$ using ancillas in the state $\ket{0}$ and CNOT gates. This gives us

\begin{align}
    &\ket{j_1}\ket{j_2}\ket{j_3}\ket{m_3}\ket{x_5}...\ket{x_n}\\
    &\ket{j_1}\ket{j_2}
\end{align}

\noindent where we write copied registers on a second line for clarity but note they are part of the same quantum state as the upper line. At this stage, we may unitarily map the upper $\ket{j_1}$ to its Yamanouchi symbol $\ket{y_1}$:

\begin{align}
    &\ket{y_1}\ket{j_2}\ket{j_3}\ket{m_3}\ket{x_5}...\ket{x_n}\\
    &\ket{j_1}\ket{j_2}.
\end{align}

\noindent We will then perform a mapping on the remaining $\ket{j_1}$ register and the upper $\ket{j_2}$ register to replace $\ket{j_2}$ with its Yamanouchi symbol $\ket{y_2}$:

\begin{align}
    &\ket{y_1}\ket{y_2}\ket{j_3}\ket{m_3}\ket{x_5}...\ket{x_n}\\
    &\ket{j_1}\ket{j_2}.
\end{align}

\noindent This is a mapping from and to computational basis states and is therefore unitary if each computational basis state in its domain is mapped to a distinct computational basis state. This is indeed the case, because given a $(j_1,j_2)$, we can deduce a unique $(j_1, y_2)$ and vice versa. We may then act on the registers $\ket{j_1}, \ket{j_2}$ and $\ket{y_2}$ to `eliminate' the register $\ket{j_1}$ (i.e. return it to $\ket{0}$), which is a unitary operation by the same argument - from any given $(j_1, j_2, y_2)$, we may deduce a unique $(j_2, y_2)$ and vice versa. This leaves us with

\begin{align}
    \ket{y_1}&\ket{y_2}\ket{j_3}\ket{m_3}\ket{x_5}...\ket{x_n}\\
             &\ket{j_2}.
\end{align}

\noindent We may then continue in this way, with each further iteration having four steps. The next iteration will be: perform the next stage of the pre-mapping to get $\ket{m_3} \mapsto \ket{j_4}$ and $\ket{x_5} \mapsto \ket{m_4}$, copy the $\ket{j_3}$ register, map the upper $\ket{j_3} \mapsto \ket{y_3}$ and then eliminate the remaining $\ket{j_2}$, resulting in

\begin{align}
    \ket{y_1}\ket{y_2}&\ket{y_3}\ket{j_4}\ket{m_4}\ket{x_6}...\ket{x_n}\\
                      &\ket{j_3}.
\end{align}

\noindent We may observe that only a logarithmic number of ancillary qubits is needed because at each step there are only $n + O(log(n))$ qubits making up our state. Eventually, we will reach

\begin{align}
    \ket{y_1}...\ket{y_{n-3}}&\ket{y_{n-2}}\ket{J}\ket{M}\\
                             &\ket{j_{n-2}}
\end{align}

\noindent and we may simply write the $\ket{j_{n-2}}$ register in the upper line to achieve the desired pre-mapping.

The Clebsch-Gordan transforms may then be done as follows. We begin with the first such transform to give

\begin{equation}
    \sum_{m_{n-2}, x_n} C^{J,M}_{j_{n-2}, m_{n-2}; 1/2, x_n}\ket{y_1}...\ket{y_{n-2}}\ket{j_{n-2}}\ket{m_{n-2}}\ket{x_n}.
\end{equation}

\noindent By acting on the $\ket{y_{n-2}}$ and $\ket{j_{n-2}}$ registers, we may then unitarily map $\ket{y_{n-2}}$ to $\ket{j_{n-3}}$ and then perform the next Clebsch-Gordan transform to produce

\begin{equation}
    \sum_{\mathclap{\substack{m_{n-2}, m_{n-3}\\x_{n-1}, x_n}}} C^{j_{n-2}, m_{n-2}}_{j_{n-3}, m_{n-3}; 1/2, x_{n-1}} \ket{y_1} ... \ket{y_{n-3}}\ket{j_{n-3}}\ket{m_{n-3}}\ket{x_{n-1}x_n}.
\end{equation}

\noindent We continue in this way until we have

\begin{equation}
    \sum C ... C\ket{y_1}\ket{j_1}\ket{j_2}\ket{m_2}\ket{x_4 ... x_n}
\end{equation}

\noindent where we have suppressed labels for readability. We may then (unitarily) eliminate $\ket{y_1}$ by acting on the registers $\ket{y_1}$ and $\ket{j_1}$ before performing the last two Clebsch-Gordan transforms to produce the desired Schur state. Notice that, as in the pre-mapping, there are only ever a constant number of registers comprised of more than one qubit - and they are comprised of logarithmically many qubits - and so throughout the procedure, logarithmically many ancillary qubits are used.

For completeness, we produce an exact count of the total number of qubits used in both the original algorithm for the Schur transform and the one of this section that uses logarithmically many ancillas. We will in fact find that, for low $n$, we are usually better off using the original algorithm, but this section's algorithm represents an exponential improvement asymptotically.

To count the number of qubits used in the original algorithm, we count the number of qubits in the state $\ket{j_1}\ket{j_2}...\ket{j_{n-2}}\ket{J}\ket{M}$ because this is a state representing the most qubits used at any point in the computation. For the same reason, to find the number of qubits used in the modified version of the algorithm in this section, we count the number of qubits in the state 

\begin{align}
    \ket{y_1}...&\ket{y_{n-3}}\ket{y_{n-2}}\ket{J}\ket{M}\\
                &\ket{j_{n-3}}\hspace{0.03cm}\ket{j_{n-2}}.
\end{align}

\noindent With this, we find that the total number of qubits used in our original algorithm for the Schur transform is 

\begin{equation}
    \sum_{k=1}^{n-1}\left\lceil log\left(\left\lfloor\frac{n+3-k}{2}\right\rfloor\right)\right\rceil + \left\lceil log(n+1)\right\rceil
\end{equation}

\noindent while the total number of qubits used in the modified algorithm that reduces the number of ancillas to logarithmic is

\begin{equation}
    n - 3 + \left\lceil log\left(\left\lfloor\frac{n}{2}\right\rfloor\right)\right\rceil + 2\left\lceil log\left(\left\lfloor\frac{n+1}{2}\right\rfloor\right)\right\rceil + \left\lceil log\left(\left\lfloor\frac{n+2}{2}\right\rfloor\right)\right\rceil + \left\lceil log\left(n+1\right)\right\rceil.
\end{equation}

\noindent Table \ref{qubitCount} contains counts for the total number of qubits used by both the original algorithm (Algorithms \ref{premapSchur} and \ref{couplingSchur}) for the Schur transform as well as the modified version of this section, for low $n$.

\renewcommand{\arraystretch}{1.5}
\begin{table}[H]

\centering
\begin{tabular}{c|c|c||c|c|c}
$n$ & Original & Modified & $n$ & Original & Modified\\\hline
     4 & 7 & 9 & 10 & 23 & 23\\
     5 & 9 & 12 & 11 & 26 & 24\\
     6 & 11 & 14 & 12 & 29 & 25\\
     7 & 13 & 15 & 13 & 32 & 26\\
     8 & 17 & 18 & 14 & 35 & 27\\
     9 & 20 & 21 & 15 & 38 & 28\\
		   
\end{tabular}
\caption{The total number of qubits used in the original algorithm for the Schur transform (Algorithms \ref{premapSchur} and \ref{couplingSchur}) as well as the modified version of this section, for low $n$. Notice that (other than an unimportant exception at $n=2$), for $n < 10$ we are better off with the original algorithm than the modification, but the modification is better asymptotically.}\label{qubitCount}
\end{table}

\section{CT-Preserving Gates}\label{CTgates}

In this section, we aim to highlight how subtle is the notion of computational tractability and, in particular, the notion of a unitary gate that `preserves CTness'. Such a set might be considered a natural notion of a `unifying classical gate set' for the study of classical simulation of quantum computations. The notion was alluded to in \cite{VdNProbMethods} but not written out as a definition. Such a definition might be written as the following.

\begin{definition}
Let $U$ be a gate acting on $k$ qubits. $U$ is called CT-preserving if, for every $n \geq k$, for every CT state $\ket{\psi}$ on $n$ qubits, $U\ket{\psi}$ is CT, where $U$ may be applied to any of the $k$ qubits of $\ket{\psi}$ (with the identity applied implicitly on all the other qubits).\label{CTpreserving}
\end{definition}

\noindent Immediately, we must note that this could never be a wholly rigorous definition, because a CT `state' really denotes an infinite family of states, whereas we refer to individual states here. However, it is common to do this in the literature and indeed \cite{VdNProbMethods} states explicitly that a CT `state' is to be taken as a state family, even when not stated explicitly. One might therefore reasonably hope that the above definition provides a well-defined notion under this unwritten assumption. Working under this assumption, we will still see that this fairly straightforward-looking definition cannot be a good one.

We start by identifying some gates that were already identified as CT-preserving in Lemma 2 of \cite{VdNProbMethods}.

\begin{definition}
Let $U$ be a gate on $n$ qubits. $U$ is called basis-preserving if for each $n$-qubit computational basis state $\ket{x}$, $U\ket{x}$ is equal to $e^{i\theta(x)}\ket{y}$ for some phase $\theta(x)$ and computational basis state $\ket{y}$.
\end{definition}

\begin{lemma}

Efficiently computable basis-preserving operations\footnote{Describing a $k$-qubit basis-preserving operation as efficiently computable means that for each $x$, $\theta(x)$ and $y$ can be classically computed in $poly(k)$ time.} are CT-preserving (\cite{VdNProbMethods}).\label{basops}
\end{lemma}
\begin{proof}
Let $U$ be a $k$-qubit efficiently computable basis-preserving operation and let $U$ act on any $k$ qubits of the computationally tractable $n$-qubit state $\ket{\psi}$. Then, for an $n$-qubit computational basis state $\ket{x}$, $\braket{x|U|\psi} = e^{i\theta(y)}\braket{y|\psi}$, where $U\ket{y} = e^{i\theta(y)}\ket{x}$, and we see $e^{i\theta(y)}\braket{y|\psi}$ may be efficiently, classically computed. Then, to sample from the set $\{|\braket{x|U|\psi}|^2 : x \in B_n\}$, we simply sample from the set $\{|\braket{y|\psi}|^2 : y \in B_n\}$ to get a sample $\hat{y}$ and conclude the answer $\hat{x}$, where $U\ket{\hat{y}} = e^{i\theta(\hat{y})}\ket{\hat{x}}$, because $|\braket{\hat{y}|\psi}|^2$ = $|\braket{\hat{x}|U|\psi}|^2$. 
\end{proof}

\noindent In some sense, the above gates are CT-preserving in a very natural way. We will now show that the Hadamard gate is also CT-preserving, although in a less natural way.

\begin{lemma}
The Hadamard gate is CT-preserving.\label{hadamard}
\end{lemma}
\begin{proof}
Suppose that we act with the Hadamard gate on the i-th qubit of the $n$-qubit computationally tractable state $\ket{\psi}$. We then find that

\begin{equation}
    \braket{x|H_i|\psi} = \frac{\braket{x|\psi}+(-1)^{x_i}\braket{x'|\psi}}{\sqrt{2}} \label{HadamardOverlap}
\end{equation}
where $x'$ is simply the bit string $x$ with the i-th bit flipped. Because we can classically, efficiently compute both $\braket{x|\psi}$ and $\braket{x'|\psi}$, we can classically, efficiently compute $\braket{x|H_i|\psi}$. We can also show that it is possible to classically and efficiently sample from $\{|\braket{x|H_i|\psi}|^2: x \in B_n\}$ in a way that is very similar to the core idea of \cite{bravyi2022simulate}. We start by sampling from $\{|\braket{x|\psi}|^2: x \in B_n\}$ to obtain an outcome $\hat{x}$. We then sample our final answer $y$ from the set $\{\hat{x}, \hat{x}'\}$ over the distribution

\begin{equation}
    p(y) = \frac{|\braket{y|H_i|\psi}|^2}{|\braket{\hat{x}|\psi}|^2 + |\braket{\hat{x}'|\psi}|^2}
\end{equation}
which is something that we are able to do classically and efficiently because we can classically and efficiently calculate the two probabilities. If we look at this sampling procedure as a whole, in order to obtain some bit string $y$, we must first obtain either $y$ or $y'$ when we sample from the set $\{|\braket{x|\psi}|^2\}$. Given this, and using the law of total probability, we find that the probability that the whole sampling algorithm samples $y$ is

\begin{equation}
    |\braket{y|\psi}|^2\frac{|\braket{y|H_i|\psi}|^2}{|\braket{y|\psi}|^2 + |\braket{y'|\psi}|^2} + |\braket{y'|\psi}|^2\frac{|\braket{y|H_i|\psi}|^2}{|\braket{y|\psi}|^2 + |\braket{y'|\psi}|^2} = |\braket{y|H_i|\psi}|^2
\end{equation}
and so we are indeed sampling from the desired distribution. This completes the proof.
\end{proof}

\noindent This is somewhat alarming at first, because, from Lemma \ref{basops}, we know that the T-gate, $\begin{pmatrix} 1 & 0 \\ 0 & e^{i\frac{\pi}{4}}\end{pmatrix}$, and CNOT are both CT-preserving. This means that the gates from the universal set $\{H, T, CNOT\}$ are all CT-preserving operations, and so it appears, under this definition, that any unitary gate may be approximated arbitrarily well by a CT-preserving gate. This is alarming because one might see this and naively conclude that all quantum states may be arbitrarily well-approximated by a CT state, which would suggest the possibility of efficient, classical simulation of all quantum computations. 

However, what saves us is the fact that we have not shown that the Hadamard gate may be applied polynomially many times to a CT state to leave a CT state. In fact, we can see in our proof where this difficulty comes in. We can see from Equation \eqref{HadamardOverlap} that each time we apply a Hadamard gate, the overlaps with computational basis states take at least twice as long to compute, in general. We therefore cannot do this any polynomial number of times. Conversely, re-examining the proof of the CT-preservation of efficiently computable basis-preserving operations, we can see that applying such a gate to a CT state, even a polynomial number of times, will leave a CT state. This separation between the `naturally' CT-preserving gates, which may be applied a polynomial number of times to preserve CTness, and the `unnatural' gates, for which this appears possible for some constant number of applications, but impossible for a general polynomial number, shows that Definition \ref{CTpreserving} is insufficient. Indeed, the idea of a gate set that preserves some set of states after a certain number of applications, but not after a certain other, greater number of applications, is nonsensical. Greater care must therefore be taken in defining a notion of a CT-preserving gate.

Even so, Lemma \ref{basops} does provide a method to simulate the application of polynomially many efficiently computable basis-preserving operations on some state from a CT state family with only a polynomial increase in the computational run time (i.e. the run time to compute computational basis overlaps and to sample from the corresponding distribution), while Lemma \ref{hadamard} does allow for the efficient simulation of a single application of a Hadamard gate to a state from a CT state family with constant overhead in the run time. Moreover, Lemma \ref{hadamard} can easily be generalised to any one qubit gate. This naturally begs the question of which gates are expected to be applicable to a state from a CT state family a polynomial number of times to only incur a polynomial increase in the relevant computational run times. We do not expect this to be any unitary gate - this would imply that we could classically and efficiently calculate an overlap of any state in the computational basis and sample from the corresponding distribution. We therefore also do not expect this set of gates to be universal. Under this reasonable complexity assumption, the question is answered for one qubit gates in the following theorem as being nothing more than the set of basis-preserving gates. We conjecture that the full set is just the basis-preserving gates also.

\begin{theorem}
Let B be the set of basis-preserving unitary gates. The set of one-qubit gates from outside B that may be combined with B to not produce a universal set is empty.
\end{theorem}

\begin{proof}
We consider rotations of the Bloch sphere:

\begin{equation}
    R_{\hat{n}}(\theta) \coloneqq e^{-i\theta\hat{n}\cdot\vec{\sigma}/2} = \cos(\theta/2)I - i\sin(\theta/2)\hat{n}\cdot\vec{\sigma}
\end{equation}

\noindent where $\vec{\sigma} = (\sigma_x, \sigma_y, \sigma_z)$. From \cite{nielsen2002quantum}, we find that one-qubit gates and CNOT are universal, and so our aim is simply to determine the one-qubit gates which, when combined with the set of one-qubit basis preserving gates, form a universal set for one-qubit unitaries. We note that contained within the set of one-qubit basis preserving gates are $\hat{z}$ rotations:

\begin{equation}
    R_{\hat{z}}(\theta) = e^{-i\theta Z/2} = \begin{pmatrix}e^{-i\theta/2} & 0 \\ 0 & e^{i\theta/2}\end{pmatrix}.
\end{equation}
From the errata of \cite{nielsen2002quantum}, we find that given any two non-parallel three-dimensional unit vectors $\hat{n}$ and $\hat{m}$, any one-qubit unitary $U$ can be written, up to a global phase, as a finite product of rotations about the $\hat{n}$ and $\hat{m}$ axis. We also find from \cite{nielsen2002quantum} that, given some rotation $R_{\hat{n}}(\theta)$, where $\theta$ is an irrational angle\footnote{An irrational angle is an irrational multiple of $2\pi$.}, repeated application of $R_{\hat{n}}(\theta)$ can be used to approximate $R_{\hat{n}}(\phi)$ arbitrarily well for any $\phi \in [0, 2\pi)$.

From these facts, we may immediately conclude that any one-qubit unitary that is, up to a global phase, a rotation by an irrational angle about an axis other than $\hat{z}$ will be universal when combined with B. Any rotation about $\hat{z}$ is in B and so it only remains to ask which rotations through rational angles about axes other than $\hat{z}$ are universal for one-qubit unitaries when combined with the set of one-qubit basis-preserving gates. Let $\hat{m} \neq \hat{z}$ and $\phi \in [0,2\pi)$ be a rational angle. We compute

\begin{multline}
    R_{\hat{z}}(\theta)R_{\hat{m}}(\phi) = \left(\cos\left(\frac{\theta}{2}\right)\cos\left(\frac{\phi}{2}\right)-\sin\left(\frac{\theta}{2}\right)\sin\left(\frac{\phi}{2}\right)\left(\hat{z}\cdot\hat{m}\right)\right)I \\- i\left(\sin\left(\frac{\theta}{2}\right)\cos\left(\frac{\phi}{2}\right)\hat{z} + \cos\left(\frac{\theta}{2}\right)\sin\left(\frac{\phi}{2}\right)\hat{m}+\sin\left(\frac{\theta}{2}\right)\sin\left(\frac{\phi}{2}\right)\hat{z}\times\hat{m}\right)\cdot\vec{\sigma}
\end{multline}
which we set equal to $R_{\tilde{n}}(\tilde{\theta})$, where 

\begin{equation}
    \cos\left(\frac{\tilde{\theta}}{2}\right) = \cos\left(\frac{\theta}{2}\right)\cos\left(\frac{\phi}{2}\right) - \sin\left(\frac{\theta}{2}\right)\sin\left(\frac{\phi}{2}\right)\left(\hat{z}\cdot\hat{m}\right)\label{newangle}
\end{equation}
The following is now useful for us (\hspace{1sp}\cite{stackexchange}):
\begin{claim}
$x$ is a rational multiple of $\pi$ $\implies$ $2\cos(x)$ is an algebraic integer.
\end{claim}
\begin{proof}
$x$ is a rational multiple of $\pi \implies (\cos(x)+i\sin(x))^n = 1$ for some $n \in \mathbb{N}$ and so $\cos(x) + i\sin(x)$ is an algebraic integer. Similarly, $\cos(x) - i\sin(x)$ is an algebraic integer. Consequently, $2\cos(x)$ is an algebraic integer.
\end{proof}
\noindent Therefore, if we can show that for some $\theta$, $2\cos\left(\frac{\tilde{\theta}}{2}\right)$ is not an algebraic integer, then $\tilde{\theta}$ is not a rational angle and, as long as $\tilde{n} \neq \hat{z}$, we are done. Algebraic integers are countable, and so it suffices to prove that the right-hand side of Equation \eqref{newangle} is non-constant as a function of $\theta$. It can only be constant if $\cos\left(\frac{\phi}{2}\right) = \hat{z}\cdot\hat{m} = 0$ which is if and only if $\phi = \pi$ and $\hat{m}$ lies on the equator.

In this case, $R_{\hat{m}}(\phi) = R_{\hat{m}}(\pi) = -i(m_1X + m_2Y)$ where $m_1^2 + m_2^2 = 1$. Up to a global phase,

\begin{align}
   R_{\hat{m}}(\phi) &= \begin{pmatrix} 0 & m_1 - im_2 \\ m_1 + im_2 & 0 \end{pmatrix}\\
   &= \begin{pmatrix} 0 & \cos\delta - i\sin\delta \\ \cos\delta + i\sin\delta & 0 \end{pmatrix}\\
   &= \begin{pmatrix} 0 & e^{-i\delta} \\ e^{i\delta} & 0 \end{pmatrix}
\end{align}

\noindent for some $\delta$. These are basis-preserving gates (and indeed, when global phases are re-introduced, these are exactly the one-qubit basis-preserving gates with off-diagonal elements).

In the case that $\phi \neq \pi$ or $\hat{z}\cdot\hat{m} \neq 0$, it remains to address the possibility that $\tilde{n} = \hat{z}$. For this, it suffices to look at the x-component of the unit vector $\tilde{n}$ which we can extract from Equation \eqref{newangle}. In order for it to be the case that $\tilde{n} = \hat{z}$, we must have

\begin{equation}
    \cos\left(\frac{\theta}{2}\right)\sin\left(\frac{\phi}{2}\right)m_1 - \sin\left(\frac{\theta}{2}\right)\sin\left(\frac{\phi}{2}\right)m_2 = 0.
\end{equation}

\noindent Assuming that the gate being added is not the identity, $\phi \neq 0$, and so $\sin\left(\frac{\phi}{2}\right)$ may be divided out of this equation. This equation is then solved only by a finite set of values $\theta \in [0, 2\pi)$. Since $\theta$ varies continuously, we conclude that as long as $\phi \neq 0 $ or $\hat{z}\cdot\hat{m} \neq 0$, we may pick a $\theta$ that makes $\tilde{\theta}$ an irrational angle and $\tilde{n} \neq \hat{z}$. Thus, given any one-qubit unitary that is not basis-preserving, it may be combined with some element of B to give a rotation of the Bloch sphere through an irrational angle about an axis other than $\hat{z}$, implying universality.
\end{proof}

\begin{conjecture}

\noindent Let $U$ be any unitary gate on $n$ qubits that is not basis-preserving. $U$ combined with the set of basis-preserving gates is universal.

\end{conjecture}

\noindent \textit{Acknowledgements}: The authors wish to thank Will Kirby and Hari Krovi for useful discussions on their implementations of the Schur transform. Gratitude is extended to Frank Verstraete for helpful suggestions regarding the section on generalised coupling. Thanks is also given to Nadish de Silva for helpful discussions on the section on CT-preservation. SS acknowledges support from the Royal Society University Research Fellowship and “Quantum Simulation Algorithms for Quantum Chromodynamics” grant (ST/W006251/1).

\bibliography{references}
\appendix
\section{Run Time Analysis for the Schur Transform and PQC Unitary Algorithms}\label{runtime}

We would like to analyse the run time (gate sequence length) for the above algorithms - for both the Schur transform and the general PQC unitaries. In both cases, we wish to find the run time of the unitaries themselves, as well as the run time of preparing individual Schur states and PQC states in the same way in which state preparation was achieved in Section \ref{statePrep}. We would also like to analyse the run time of the modified Schur transform algorithm of Section \ref{logManyAncillas} that uses logarithmically many ancillas. It will not be shown explicitly, but it can be easily checked from Claim \ref{schurTransformRuntime} that the run time of the modified version of the algorithm is asymptotically the same as the original version.

To analyse the run time, we use the same ideas for run time analysis as the authors in \cite{kirby2017practical}. The idea is to first decompose the algorithm into a sequence of two-level rotations, whose length is the `two-level gate sequence length'. Each of these may then be decomposed further into gates from the Clifford + T universal gate set, which may be implemented in a fault-tolerant manner. The length of the final sequence of Clifford + T gates is the `Clifford + T gate sequence length'. 

To translate the two-level gate sequence length to the Clifford + T gate sequence length, we use the same result as \cite{kirby2017practical} that a two-level rotation on $n$ qubits may be decomposed to an accuracy of $\delta$ (in the trace norm) in $O(n\;log(1/\delta))$ Clifford + T operations. If our two-level gate sequence length is p, given that each of our two-level rotations acts on $O(log(n))$ qubits, our Clifford + T gate sequence length is $O(p\;log(n)\;log(p/\epsilon))$, where $\epsilon$ is the desired final accuracy of our implementation in terms of the trace norm. 

With this in mind, we present each gate sequence length for each of the algorithms in question, all of which will be proved below.
\renewcommand{\arraystretch}{1.8}
\begin{table}[h]

\centering
\begin{tabular}{c|c|c||c|c}
&\multicolumn{2}{c||}{Schur}&\multicolumn{2}{c}{General PQC}\\\hline
     &Two-Level&Clifford + T & Two-Level & Clifford + T\\\hline
     Unitary & $O(n^3)$&$O(n^3log(n)log(\frac{n}{\epsilon}))$&$O(n^6)$&$O(n^6log(n)log(\frac{n}{\epsilon}))$\\
     \makecell{Individual State\\ Preparation} & $O(n^2)$ & $O(n^2log(n)log(\frac{n}{\epsilon}))$&$O(n^3)$&$O(n^3log(n)log(\frac{n}{\epsilon}))$
		   
\end{tabular}\label{runtimes}
\caption{Gate sequence lengths. The gate sequence lengths shown for the Schur unitary apply to both our original algorithm and the modified version of Section \ref{logManyAncillas} that uses logarithmically many ancillas. This will not be proved but can be easily checked.}
\end{table}

Again in \cite{kirby2017practical}, it is also used that a unitary $U$ may be decomposed into a sequence of two-level rotations whose length equals the number of non-zero elements on or below the main diagonal of U, not counting 1's on the main diagonal. In what follows, we will bound this above simply by the number of non-zero entries in U, not counting 1's on the main diagonal.

We will now prove all of the above two-level run times, which immediately give us the Clifford + T run times.

\begin{claim}
Our algorithm for the Schur transform may be decomposed into a sequence of $O(n^3)$ two-level rotations.\label{schurTransformRuntime}
\end{claim}

\begin{proof}
The first stage, the pre-mapping, consists of $O(n)$ steps, which may generically be represented as $\ket{j_k}\ket{m_k}\ket{x_{k+2}} \mapsto \ket{j_k}\ket{j_{k+1}}\ket{m_{k+1}}$. The registers on the left-hand side take $O(n^2)$ values and they are each mapped to an individual computational basis state. Thus, the two-level gate sequence length of the pre-mapping stage is $O(n^3)$.

In the coupling stage, there are again $O(n)$ steps, where we may write a generic one as $\ket{j_k}\ket{j_{k+1}}\ket{m_{k+1}} \mapsto \sum_{m_k, x_{k+2}} C^{j_{k+1}, m_{k+1}}_{j_k, m_k; 1/2, x_{k+2}}\ket{j_k}\ket{m_k}\ket{x_{k+2}}$. Notice that the sum on the right-hand side may run over at most 2 elements, because of conservation of angular momentum and $x_{k+2} \in \left\{\pm \frac{1}{2}\right\}$. The registers on the left-hand side may again run over $O(n^2)$ values, because $j_k \leq \frac{n}{2}$, $|j_{k+1} - j_k| = \frac{1}{2}$ and $m_{k+1} \in \{-j_{k+1}, -j_{k+1} + 1, ..., j_{k+1}\}$, which means that the two-level gate sequence length of the coupling stage is also $O(n^3)$.
\end{proof}

\begin{claim}
Our algorithm for individual state preparation, when applied to Schur states, may be decomposed into a sequence of $O(n^2)$ two-level rotations.
\end{claim}

\begin{proof}
As in Section \ref{genCoupling}, we may start with labelled individual qubits in any state, say $\ket{0}$:

\begin{equation}
    \ket{j_1}...\ket{j_{n-2}}\ket{J}\ket{M}
\end{equation}

\noindent so that each spin eigenvalue is fixed to its target value. Then, a generic operation looks like $\ket{j_k}\ket{j_{k+1}}\ket{m_{k+1}} \mapsto \sum_{m_k, x_{k+2}} C^{j_{k+1}, m_{k+1}}_{j_k, m_k; 1/2, x_{k+2}}\ket{j_k}\ket{m_k}\ket{x_{k+2}}$. The registers on the left-hand side may take $O(n)$ values, and again, the sum on the right-hand side runs over at most 2 values. Since we make $O(n)$ of these operations, our final two-level gate sequence length is $O(n^2)$.
\end{proof}

\begin{claim}
Our algorithm for the implementation of a general PQC unitary may be decomposed into a sequence of $O(n^6)$ two-level rotations.
\end{claim}

\begin{proof}
Each of the $O(n)$ steps of the pre-mapping stage may be written as $\ket{j_{l(i)}}\ket{m_{l(i)}}\ket{j_{r(i)}}\ket{m_{r(i)}} \mapsto$\newline $\ket{j_{l(i)}}\ket{j_{r(i)}}\ket{j_i}\ket{m_i}$. The registers of the left-hand side may take $O(n^4)$ values, and each map to an individual computational basis state, meaning that the pre-mapping has two-level gate sequence length $O(n^5)$. 

Then, the coupling stage comprises $O(n)$ steps of the form $\ket{j_{l(i)}}\ket{j_{r(i)}}\ket{j_i}\ket{m_i} \mapsto$\newline$ \sum_{m_{l(i)}, m_{r(i)}}C^{j_i, m_i}_{j_{l(i)}, m_{l(i)}; j_{r(i)}, m_{r(i)}}\ket{j_{l(i)}}\ket{m_{l(i)}}\ket{j_{r(i)}}\ket{m_{r(i)}}$. The left-hand registers may again take $O(n^4)$ values, but now the sum on the right-hand side runs over $O(n)$ elements, given conservation of angular momentum. Each step therefore has a two-level gate sequence length of $O(n^5)$, giving the whole stage a two-level gate sequence length of $O(n^6)$.
\end{proof}

\begin{claim}
Our algorithm for individual state preparation, when applied to general PQC states, may be decomposed into a sequence of $O(n^3)$ two-level rotations.
\end{claim}

\begin{proof}
Again, we may start with $n$ individual qubits labelled as

\begin{equation}
\ket{j_1}...\ket{j_{n-2}}\ket{J}\ket{M}
\end{equation}

\noindent and operate $O(n)$ times as $\ket{j_{l(i)}}\ket{j_{r(i)}}\ket{j_i}\ket{m_i} \mapsto \sum_{m_{l(i)}, m_{r(i)}}C^{j_i, m_i}_{j_{l(i)}, m_{l(i)}; j_{r(i)}, m_{r(i)}}\ket{j_{l(i)}}\ket{m_{l(i)}}\ket{j_{r(i)}}\ket{m_{r(i)}}$, where all the $j_i$ are fixed to their target values, as always. The left-hand registers may therefore take $O(n)$ values and the right-hand sum runs over $O(n)$ values, again by conservation of angular momentum. Each operation therefore has a two-level gate sequence length of $O(n^2)$ and so the whole algorithm has a two-level gate sequence length of $O(n^3)$.
\end{proof}

\section{The Necessity of the Schur Transform as a Unitary Operation}\label{cleanTransform}

In this work, we presented algorithms for the quantum Schur transform as a unitary operation on a quantum computer. We now discuss where the differences lie in the differing notions of the Schur transform that exist, namely the operation as an isometry versus as a unitary, and why the unitary implementation (which results from our `pre-mapping' stage, see Section \ref{schurAlgoSection}) is essential in some applications.

In Section \ref{PQCUnitaries}, we defined Schur states as $n$-qubit states labelled $\ket{j_1, ..., j_{n-2},J,M}$ having total angular momentum (spin) $j_1$ on their first two qubits, $j_2$ on their first three qubits, and up to spin $j_{n-2}$ on their first $n-1$ qubits, with total spin $J$ on all their qubits. Finally, this state has $Z$-angular momentum $M$ on all its qubits. In terms of $SU(2)$ Clebsch-Gordan coefficients, this state may be written as in Equation \eqref{SchurState}:

\begin{equation}
    \ket{j_1, j_2, ..., j_{n-2},J,M} = \sum_{x \in B_n} C^{j_1, x_1 + x_2}_{1/2, x_1; 1/2, x_2} C^{j_2, x_1 + x_2 + x_3}_{j_1, x_1 + x_2; 1/2, x_3} ... C^{J, M}_{j_{n-2}, x_1 + ... + x_{n-1}; 1/2, x_n} \ket{x}
\end{equation}
where again, for convenience, we sum over the set $B_n = \{\pm\frac{1}{2}\}^n$ and label computational basis states by $\ket{\pm \frac{1}{2}}$. As is common in more physical conventions, we defined the Schur transform as the unitary operation rotating the computational basis on $n$ qubits, $\{\ket{x}\}_{x \in B_n}$, to the Schur basis on $n$ qubits, $\ket{j_1, ..., j_{n-2}, J,M}$, with no particular order specified (the user may specify an order if desired). We recall that we achieved this in Section \ref{schurAlgoSection} in two stages - first was the pre-mapping stage which performs\footnote{Here, we only discuss the first version of the algorithm with $\mathcal{O}(n\log(n))$ ancillary qubits, as opposed to the modified version with logarithmically many qubits, for simplicity. Note that the same discussion goes through in this case.}

\begin{equation}
    \ket{x} \mapsto \ket{j_1}\ket{j_2}...\ket{j_{n-2}}\ket{J}\ket{M}
\end{equation}
where, again, the right-hand side is a computational basis state on $\mathcal{O}(n\log(n))$ qubits. Each ket $\ket{j_i}$, $\ket{J}$ and $\ket{M}$ separately encodes a spin eigenvalue and the sequence of eigenvalues on the right-hand side must be allowed by the rules of angular momentum - in particular the $j$-values must differ by $\pm 1/2$ from their neighbours and must be non-negative, while we must have $M \in \{-J, -J+1, ..., J\}$\footnote{The generalisation to qudits and $SU(d)$ representation theory is straightforward - see Section \ref{SUN}.}. There are $2^n$ possibilities for such valid sequences of eigenvalues, so such a mapping is indeed possible. No particular order is specified here, but one could use a canonical one, for example that given by the RSK correspondence \cite{krattenthaler2006growth}.

The second stage of our algorithm, to complete the Schur transform unitary, is the coupling stage:

\begin{equation}
    \ket{j_1}\ket{j_2}...\ket{j_{n-2}}\ket{J}\ket{M} \mapsto \ket{j_1, j_2, ..., j_{n-2},J,M}\label{couplingSchurStage}
\end{equation}
where we emphasise that in this operation ancillary qubits have become unentangled.

An alternative definition of the Schur transform, as is considered in, for example, \cite{harrow2005applications,kirby2017practical}, is the operation performing

\begin{equation}
    \ket{j_1, ..., j_{n-2},J,M} \mapsto \ket{j_1}\ket{j_2}...\ket{j_{n-2}}\ket{J}\ket{M}\label{mathematicalSchur}.
\end{equation}
This is the definition that is not adopted in this work. The history of this definition may be traced back to more mathematical representation theoretic work and so one might refer to the two notions as the `physical' and `mathematical' notions respectively. Sometimes, also, the unitary operation is referred to as a `clean' transform, given that the ancillary qubits have become unentangled by the end. There are two points that can lead to confusion between these two definitions. The first and most important is that the physical notion is a unitary operation, mapping $n$ qubits into $n$ qubits, whereas the mathematical notion is an isometry, mapping $n$ qubits into $\mathcal{O}(n\log(n))$ qubits. Secondly, the mathematical Schur transform is akin to the inverse of the physical Schur transform. Indeed, one can see that the mathematical definition of the Schur transform presented in Equation \eqref{mathematicalSchur} is exactly the inverse of the coupling stage in Equation \eqref{couplingSchurStage}.

For many applications, the mathematical Schur transform suffices and, accordingly, most of the literature on the Schur transform has been implementing only this isometry. However, for some ``PQC-like'' applications, for example the algorithms of \cite{zheng2022super}, the unitary is a necessity, hence the emphasis placed on unitarity in the present work.

Let us consider why it is necessary to perform the unitary in some applications. Indeed, suppose one wished to perform the inverse (unitary, physical) Schur transform, perform some unitary gate (say, a Hadamard on the first qubit) and then perform the forwards (unitary, physical) Schur transform. One notes that, in the same way there are $2^n$ computational basis states $\ket{x}$, there are $2^n$ valid Schur encodings $\ket{j_1}\ket{j_2}...\ket{j_{n-2}}\ket{J}\ket{M}$, and so could we not just perform the gate $H \otimes I^{\otimes (n-1)}$ on the $2^n$ valid computational basis states $\ket{j_1}\ket{j_2}...\ket{j_{n-2}}\ket{J}\ket{M}$? In theory, this could be done, but it is not clear \textit{a priori} how it can be done efficiently. Indeed, one must check each of the $\mathcal{O}(n\log(n))$ qubits to see if a given computational basis state forms a valid Schur encoding and so one \textit{a priori} has to act on every qubit simultaneously, which is inefficient. Thus, in this situation, the pre-mapping stage exactly gives us the ability to efficiently package the $2^n$ valid $\ket{j_1}\ket{j_2}...\ket{j_{n-2}}\ket{J}\ket{M}$ into the $2^n$ computational basis states $\ket{x}$ so that we can operate on them. More generally, for any algorithm requiring a Schur transform followed by an $n$-qubit unitary, the implementation of the Schur transform as a unitary is a necessity.

We finish by mentioning one other small point of confusion, although this is easily alleviated, and that is differences in notation. In the mathematical context, Schur states $\ket{j_1, j_2, ..., j_{n-2}, J, M}$ may be denoted as, for example $\ket{\lambda, q_\lambda, p_\lambda}$, whereas their encodings $\ket{j_1}\ket{j_2}...\ket{j_{n-2}}\ket{J}\ket{M}$ may be denoted $\ket{\lambda}\ket{q_\lambda}\ket{p_\lambda}$, as is seen in, for example, \cite{bacon2005quantum}. The total angular momentum (spin) of all $n$ qubits, $J$, corresponds to the partition of $n$, $\lambda$. $q_\lambda$, indexing within irreps of the unitary group, corresponds to $M$, and lastly $p_\lambda$, the index labelling states within irreps of the symmetric group, corresponds to all internal $j$ spins $(j_1, ..., j_{n-2})$.

\end{document}